\documentclass[%
  aps,%
  prx,%
  letterpaper,
  superscriptaddress,%
  reprint,%
  longbibliography,%
]{revtex4-2}

\usepackage[preset=reset,pkgset=extended,hyperrefdefs={defer,noemail}]{phfnote}
\usepackage{phfqit}
\usepackage{phfparen}

\usepackage[thmset=rich,smallproofs=false,proofref={marginbottom}]{phfthm}

\def\proofrefautopageref#1{\hyperref[#1]{\,\guilsinglright\;p.\,\pageref*{#1}}}%
\renewcommand\proofrefsize[1]{%
  \begingroup
  \tiny%
  \fboxsep=1pt\relax
  \let\autopageref\proofrefautopageref
    \parbox[][][c]{3em}{#1}%
  \endgroup
}

\usepackage[reset]{geometry}
\makeatletter\Gm@restore@org\makeatother

\usepackage{ragged2e}
\DeclareCaptionJustification{myjust}{\justify}
\usepackage[semibold]{sourcesanspro}

\usepackage{phfcc}
\phfMakeCommentingCommand[initials={PhF},color={phfcc5}]{phf}

\usepackage{hyperref}
\usepackage[capitalize,nameinlink]{cleveref}

\DeclarePairedDelimiterXPP\opnorm[1]{}{\lVert}{\rVert}{}{{#1}}
\DeclarePairedDelimiterXPP\onenorm[1]{}{\lVert}{\rVert}{_1}{{#1}}
\DeclarePairedDelimiterXPP\twonorm[1]{}{\lVert}{\rVert}{_2}{{#1}}
\DeclarePairedDelimiterXPP\dianorm[1]{}{\lVert}{\rVert}{_\diamond}{{#1}}

\let\leq\leqslant
\let\geq\geqslant

\newcounter{maintheorem}
\phfMakeTheorem[counter=maintheorem,aliascounter=false,proofref=false]{maintheorem}{Theorem}

\makeatletter
\newcommand\bibalias[2]{%
  \@namedef{bibali@#1}{#2}%
}

\newcommand\biba@deblank[1]{\romannumeral\biba@deblank@#1/ /} 
\long\def\biba@deblank@#1 /{\biba@deblank@i#1/}
\long\def\biba@deblank@i#1/#2{\z@#1}

\newtoks\biba@toks
\let\bibalias@oldcite\cite
\def\cite{%
  \@ifnextchar[{%
    \biba@cite@optarg%
  }{%
    \biba@cite{}%
  }%
}
\newcommand\biba@cite@optarg[2][]{%
  \biba@cite{[#1]}{#2}%
}
\newcommand\biba@cite[2]{%
  \biba@toks{\bibalias@oldcite#1}%
  \def\biba@comma{}%
  \def\biba@all{}%
  \def\biba@argkeys{}%
  \@for\biba@one@:=#2\do{%
    \edef\biba@one{\expandafter\@firstofone\biba@one@\@empty}%
    \edef\biba@one{\expandafter\biba@deblank\expandafter{\biba@one}}
    \edef\biba@argkeys{\biba@argkeys\biba@comma\biba@one}%
    \@ifundefined{bibali@\biba@one}{%
      \edef\biba@all{\biba@all\biba@comma\biba@one}%
    }{%
      \PackageInfo{bibalias}{%
        Replacing citation `\biba@one' with `\@nameuse{bibali@\biba@one}'
      }%
      \edef\biba@all{\biba@all\biba@comma\@nameuse{bibali@\biba@one}}%
    }%
    \def\biba@comma{,}%
  }%
  \immediate\write\@auxout{\noexpand\bgroup\noexpand\renewcommand\noexpand\citation[1]{}\noexpand\citation{\biba@argkeys}\noexpand\egroup}%
  \edef\biba@tmp{\the\biba@toks{\biba@all}}%
  \biba@tmp%
}
\makeatother

\bibalias{Albert2019arXiv_robust}{Albert2020PRX_robust}
\bibalias{Alhambra2018arXiv_algorithmic}{Alhambra2019Qu_algorithmic}
\bibalias{Anshu2017arXiv_compression}{Anshu2019IEEETIT_compression}
\bibalias{Anshu2017arXiv_oneshot}{Anshu2019IEEETIT_oneshot}
\bibalias{Anshu2018arXiv_partially}{Anshu2020ITIT_partially}
\bibalias{ArnonFriedman2023arXiv_computational}{Arnon-Friedman2023arXiv_computational}
\bibalias{Balasubramanian2022arXiv_new}{Balasubramanian2022PRD_chaos}
\bibalias{Berta2018arXiv_amortized}{Wilde2020LMP_amortized}
\bibalias{Bertoni2020arXiv_algebraic}{Bertoni2020NJP_algebraic}
\bibalias{Biamonte2016arXiv_machinelearning}{Biamonte2017Nat_machinelearning}
\bibalias{Bjelakovic2003arXiv_compression}{Bjelakovic2005QIP_compression}
\bibalias{BookHan02_InfSpecMethods}{BookHan_InfSpecMethods}
\bibalias{BookNielsenChuang2000}{BookNielsenChuang2010}
\bibalias{BookZeng2015arXiv_matter}{BookZeng2019_matter}
\bibalias{Brandao2012CMP_local}{Brandao2016CMP_local}
\bibalias{Brandao2017arXiv_chainAQECC}{Brandao2019PRL_chainAQECC}
\bibalias{Brandao2019arXiv_models}{Brandao2021PRXQ_models}
\bibalias{Bravyi2017arXiv_shallow}{Bravyi2018Sci_shallow}
\bibalias{Breuckmann2021arXiv_LDPC}{Breuckmann2021PRXQ_LDPC}
\bibalias{Brown2013_decoupling}{Brown2015CMP_decoupling}
\bibalias{Brown2019arXiv_nonClifford}{Brown2020SciAdv_nonClifford}
\bibalias{Buscemi2019arXiv_dichotomies}{Buscemi2019Qu_dichotomies}
\bibalias{Camargo2020arXiv_entanglement}{Camargo2021PRR_entanglement}
\bibalias{Capel2018arXiv_conditional}{Capel2018JPA_conditional}
\bibalias{Chitambar2018arXiv_resource}{Chitambar2019RMP_resource}
\bibalias{Chubb2017arXiv_beyond}{Chubb2018Qu_beyond}
\bibalias{Cirstoiu2019arXiv_robustness}{Cirstoiu2020PRX_robustness}
\bibalias{Correa2016arXiv_thermometry}{Correa2017PRA_thermometry}
\bibalias{Dupuis2016arXiv_accumulation}{Dupuis2020CMP_accumulation}
\bibalias{Dupuis2018arXiv_improved}{Dupuis2019IEEETIT_improved}
\bibalias{Eisert2021arXiv_entangling}{Eisert2021PRL_entangling}
\bibalias{Faist2018arXiv_thcapacity}{Faist2019PRL_thcapacity}
\bibalias{Faist2019arXiv_continuous}{Faist2020PRX_continuous}
\bibalias{Faist2019arXiv_impl}{Faist2021CMP_impl}
\bibalias{Faist2022arXiv_timeenergy}{Faist2023PRXQ_timeenergy}
\bibalias{Fang2018arXiv_simulation}{Fang2019IEEETIT_simulation}
\bibalias{Garcia2022arXiv_resource}{Garcia2023PNAS_resource}
\bibalias{Gorecki2019arXiv_multiparameter}{Gorecki2020Qu_multiparameter}
\bibalias{Gottesman1997PhD_stab}{PhDGottesman1997_stab}
\bibalias{Gour2017arXiv_entropic}{Gour2018NatComm_entropic}
\bibalias{Gour2018arXiv_entropychannel}{Gour2021PRR_entropychannel}
\bibalias{Gschwendtner2019arXiv_lowenergies}{Gschwendtner2019JHEP_lowenergies}
\bibalias{Haferkamp2020arXiv_improved}{Haferkamp2021PRA_improved}
\bibalias{Haferkamp2021arXiv_linear}{Haferkamp2022NPhys_linear}
\bibalias{Halpern2021arXiv_transport}{YungerHalpern2022nQI_transport}
\bibalias{Hangleiter2023arXiv_computational}{Hangleiter2023RMP_computational}
\bibalias{Harlow2018arXiv_constraints}{Harlow2019PRL_constraints}
\bibalias{Hayden2017arXiv_frame}{Hayden2021PRXQ_frame}
\bibalias{HindsMingo2018arXiv_multiple}{HindsMingo2018_multiple}
\bibalias{HindsMingo2018arXiv_superpositions}{HindsMingo2019Qu_superpositions}
\bibalias{HunterJones2017JHEP_random}{HunterJones2018JHEP_random}
\bibalias{Jennings2010PRL_arrow}{Jennings2010PRE_arrow}
\bibalias{Ji2017arXiv_pseudorandom}{Ji2018CRYPTO_pseudorandom}
\bibalias{Katariya2020arXiv_geometric}{Katariya2021QIP_geometric}
\bibalias{Kato2016arXiv_Markov}{Kato2019CMP_Markov}
\bibalias{Khandelwal2019arXiv_general}{Khandelwal2020Qu_general}
\bibalias{Kubica2020arXiv_metrological}{Kubica2021PRL_metrological}
\bibalias{Liu2021arXiv_approximate}{Liu2023nQI_approximate}
\bibalias{Lostaglio2018arXiv_broadcast}{Lostaglio2019PRL_broadcast}
\bibalias{Manzano2020arXiv_nonabelian}{Manzano2022PRXQ_nonabelian}
\bibalias{Marvian2018arXiv_distillation}{Marvian2020NC_distillation}
\bibalias{Mele2023arXiv_tutorial}{Mele2023arXiv_introduction}
\bibalias{Munson2024arXiv_cplx}{Munson2025PRXQ_cplx}
\bibalias{Nahum2018PRX_operator}{Nahum2018PRX_spreading}
\bibalias{Ng2018arXiv_TO}{Ng2018TiQRbook_resource}
\bibalias{Ouyang2019arXiv_robust}{Ouyang2022IEEETIT_robust}
\bibalias{Piveteau2021arXiv_error}{Piveteau2021PRL_error}
\bibalias{Popescu2018arXiv_applications}{Popescu2018PTRSA_applications}
\bibalias{Rojkov2021arXiv_bias}{Rojkov2022PRL_bias}
\bibalias{Rouze2021arXiv_learning}{Rouze2024Q_learning}
\bibalias{Sagawa2019arXiv_asymptotic}{Sagawa2021JPA_asymptotic}
\bibalias{Shettell2021arXiv_limits}{Shettell2021NJP_limits}
\bibalias{Sparaciari2016arXiv_workheat}{Sparaciari2017PRA_workheat}
\bibalias{Sparaciari2018arXiv_first}{Sparaciari2020Qu_first}
\bibalias{Susskind2018PiTP_complexity1}{Susskind2018PiTP_three}
\bibalias{Suzuki2023arXiv_cplx}{Suzuki2025Qu_cplx}
\bibalias{Taranto2021arXiv_cooling}{Taranto2023PRXQ_landauer}
\bibalias{TownsendTeague2023arXiv_floquetifying}{TownsendTeague2023EPTCS_floquetifying}
\bibalias{Wehner2015arXiv_reversibility}{Alhambra2018PRA_reversibility}
\bibalias{Weilenmann2018arXiv_smooth}{Weilenmann2018_smooth}
\bibalias{White2022arXiv_manytime}{White2022arXiv_manybody}
\bibalias{Winter2015arXiv_tight}{Winter2016CMP_tight}
\bibalias{Witten2018arXiv_entanglement}{Witten2018RMP_entanglement}
\bibalias{Woods2015arXiv_more}{Woods2019Q_more}
\bibalias{Woods2019arXiv_continuous}{Woods2020Qu_continuous}
\bibalias{Xu2022arXiv_qubitoscillator}{Xu2023PRXQ_qubitoscillator}
\bibalias{Yang2020arXiv_covariant}{Yang2022PRR_covariant}
\bibalias{YungerHalpern2018arXiv_photoisomerization}{YungerHalpern2020PRA_photoisomerization}
\bibalias{YungerHalpern2019arXiv_equilibration}{YungerHalpern2020PRE_equilibration}
\bibalias{YungerHalpern2021arXiv_resource}{YungerHalpern2022PRA_resource}
\bibalias{Zhou2019arXiv_strongly}{Zhou2021PRX_strongly}
\bibalias{Zhou2020arXiv_covariant}{Zhou2021Qu_covariant}

\bibalias{arXiv:1911.05563}{Faist2021CMP_impl}

\begin{document}
\title{Universal thermodynamic implementation of a process with a variable work cost}
\author{Philippe Faist}
\affiliation{Dahlem
  Center for Complex Quantum Systems, Freie Universität Berlin, 14195 Berlin,
  Germany}
\date{Jan 22, 2026}
\begin{abstract}
  The minimum amount of thermodynamic work required in order to implement a
  quantum computation or a quantum state transformation can be quantified using
  frameworks based on the resource theory of thermodynamics, deeply rooted in
  the works of Landauer and Bennett.  For instance, the work we need to invest
  in order to implement $n$ independent and identically distributed (i.i.d.)
  copies of a quantum channel is quantified by the \emph{thermodynamic capacity}
  of the channel when we require the implementation's accuracy to be guaranteed
  in diamond norm over the $n$-system input.
  Recent work showed that work extraction can be implemented \emph{universally},
  meaning the same implementation works for a large class of input states, while
  achieving a \emph{variable work cost} that is optimal for each individual
  i.i.d.\@ input state.
  Here, we revisit some techniques leading to derivation of the thermodynamic
  capacity, and leverage them to construct a thermodynamic implementation of $n$
  i.i.d.\@ copies of any time-covariant quantum channel, up to some process
  decoherence that is necessary because the implementation reveals the amount of
  consumed work.  The protocol uses so-called \emph{thermal operations} and
  achieves the optimal per-input work cost for any i.i.d.\@ input state; it
  relies on the conditional erasure protocol in our earlier work, adjusted to
  yield variable work.
  We discuss the effect of the work-cost decoherence.  While it can
  significantly corrupt the correlations between the output state and any
  reference system, we show that for any time-covariant i.i.d.\@ input state,
  the state on the output system faithfully reproduces that of the desired
  process to be implemented.
  As an immediate consequence of our results, we recover recent results for
  optimal work extraction from i.i.d.\@ states up to the error scaling and
  implementation specifics, and propose an optimal preparation protocol for
  time-covariant i.i.d.\@ states.
\end{abstract}
\maketitle

Ever since the founding works of Landauer and Bennett on the thermodynamics of
computation~\cite{%
  Landauer1961_5392446Erasure,Landauer1996_physical,%
  Bennett1982IJTP_ThermodynOfComp,Bennett1973IBMJRD_LogRevComp}, a rich line of
research has quantified fundamental restrictions on possible state
transformations~\cite{%
  BookBinder2018_ThermoQuantumRegime,Goold2016JPA_review,%
  BookSagawa2022_entropy,%
  Sagawa2009PRL_minimal,%
  Horodecki2013_ThermoMaj,Brandao2013_resource,%
  YungerHalpern2016NatComm_NATSandNATO,%
  Brandao2015PNAS_secondlaws,Lostaglio2017NJP_noncommutativity} as well as the
minimal thermodynamic resources, quantified in terms of thermodynamic work,
required in order to perform tasks such as work extraction, state preparation,
and to implement any quantum
channel~\cite{Dahlsten2011NJP_inadequacy,delRio2011Nature,Aberg2013_worklike,%
  Faist2015NatComm,PhDPhF2016,Gour2018NatComm_entropic,Faist2018PRX_workcost}.
To implement thermodynamic tasks, such as the extraction of thermodynamic work
from a state $\rho$ that is out of equilibrium, one can appeal to the framework
of \emph{thermal
  operations}~\cite{Horodecki2013_ThermoMaj,Brandao2013_resource}, in which
arbitrary energy-conserving unitary interactions with a heat bath are allowed.
To derive fundamental limits, it is convenient to use the framework of
\emph{Gibbs-preserving maps}, where one derives bounds that apply even to an
agent who could carry out for free any arbitrary quantum channel that fixes the
thermal state~\cite{Faist2015NatComm,Faist2018PRX_workcost}.  Bounds derived
with Gibbs-preserving maps are often conveniently formulated as semidefinite
programs, and they automatically apply in the context of thermal operations.
The gap between both types of operations can be traced back to the fact that
thermal operations are time-covariant, meaning that they commute with time
evolution~\cite{Faist2015NJP_Gibbs,%
  Gour2018NatComm_entropic,Marvian2014PRA_modes,%
  Lostaglio2015NC_beyond,Lostaglio2015PRX_coherence,Faist2015NJP_Gibbs,%
  Tajima2024arXiv_gibbspreserving,Tajima2025arXiv_universal}.

Rather than a state transformation $\rho\to\sigma$, one can quantify the work
cost of implementing a particular process, specified by a quantum channel
$\mathcal{E}_X$.  If we care only about successfully implementing
$\mathcal{E}_X$ on a fixed input state $\sigma_X$, we might demand that the
implementation correctly prepare the output state $\mathcal{E}_X(\sigma_{XR})$
for any purification $\ket\sigma_{XR}$ of the fixed input
$\sigma_X$~\cite{Faist2015NatComm}.  (Reproducing the correlations between the
output system and $R$ ensures that the process $\mathcal{E}$ was implemented on
$\sigma_X$, and can carry a different work cost than implementing the simple
state transformation
$\sigma_X \to \mathcal{E}(\sigma_X)$~\cite{Faist2015NatComm,PhDPhF2016}.)  

For small systems and if $\sigma_X$ is sufficiently mixed, an implementation
that correctly prepares $\mathcal{E}_X(\sigma_{XR})$ when given $\sigma_{XR}$ as
an input is also guaranteed to implement $\mathcal{E}_X$ accurately on other
input states.  This follows from the fact that $\mathcal{E}(\sigma_{XR})$ fully
specifies the Choi matrix of $\mathcal{E}$.  This argument only applies to small
systems and for sufficiently mixed $\sigma_X$, because the accuracy to which the
Choi state is specified depends on dimensional factors and the minimal
eigenvalue of $\sigma_X$.

We might, instead, demand that an implementation of $\mathcal{E}_X$ reproduce
the correct output for all input states, with a guaranteed accuracy in diamond
norm.  This condition guarantees the implementation prepares
$\mathcal{E}_X(\sigma_{XR})$ accurately, for any input $\sigma_{XR}$.  Such
implementations are called \emph{universal}.
An asymptotically optimal universal implementation of an i.i.d.\@ channel
$\mathcal{E}^{\otimes n}$, with a deterministic work cost, was constructed in
Refs.~\cite{Faist2019PRL_thcapacity,Faist2021CMP_impl}.  Its work cost is given
by the channel's \emph{thermodynamic capacity}; it is equal to the worst case
over all individual work costs of input-state-specific implementations for all
possible i.i.d.\@ input states.  The proof relies on Schur-Weyl
duality~\cite{PhDHarrow2005,Haah2017IEEETIT_sampleoptimal} and bears resemblance
to the quantum reverse Shannon theorem's
proof~\cite{Bennett2014_reverse,Berta2011_reverse}.
Our Ref.~\cite{Faist2021CMP_impl} presents its main theorems in the case where
universality is requested over all input states, but its proof techniques extend
immediately to universality over any arbitrary set of input states.  This fact
could have been better communicated in
Refs.~\cite{Faist2019PRL_thcapacity,Faist2021CMP_impl}, and we append complete
statements to this effect in \cref{appx:UniversalWithSetOfInputs}.
By considering the process $\mathcal{E}_X(\cdot) = \tr(\cdot)\,\gamma_X$ which
deterministically prepares the output thermal state $\gamma_X$ at the background
inverse temperature $\beta>0$, and any arbitrary set of states
$\hat{\mathscr{S}}_X$, these results give a universal work extraction scheme
from i.i.d.\@ states with an optimal deterministic extracted work per copy
$\beta^{-1}\min_{\sigma_X\in\hat{\mathscr{S}}_X}\DD{\sigma_X}{\gamma_X}$, where
$\DD{\sigma}{\sigma'} = \tr`\big[\sigma`\big(\log\sigma-\log\sigma')]$ is the
quantum relative entropy.  This work cost coincides with the worst-case free
energy over inputs from the given set.

Recently, a series of works have studied in closer detail how the specific task
of work extraction can be implemented
universally~\cite{Safranek2023PRL_work,Chakraborty2025QST_sample,%
  Watanabe2024PRL_black,Watanabe2025arXiv_universal}.
In particular, Ref.~\cite{Watanabe2025arXiv_universal} presents a work
extraction protocol with a variable work cost, rather than a deterministic work
cost.  The protocol of Ref.~\cite{Watanabe2025arXiv_universal} is universal over
all i.i.d.\@ input states.  Namely, when given as an input
$\sigma_X^{\otimes n}$, it extracts an asymptotic amount of work per copy
$\beta^{-1}\DD{\sigma_X}{\gamma_X}$, which is optimal even for a protocol that
would have been specifically designed to operate only on $\sigma_X^{\otimes n}$.

In this contribution, we revisit the techniques of our earlier
Refs.~\cite{Faist2019PRL_thcapacity,Faist2021CMP_impl}, and show that they can
lead to universal, variable-work cost implementations of a large class of
time-covariant quantum processes, including work extraction as well as the
preparation of time-covariant states.
Specifically, given an arbitrary process $\mathcal{E}^{\otimes n}_{X\to X}$ in
which $\mathcal{E}_{X\to X}$ commutes with the Hamiltonian time evolution (i.e.,
it is \emph{time-covariant}), we exhibit a protocol with variable work cost that
implements a map $\mathcal{T}_{X^n\to X^n}$ with the following properties.
(i)~\emph{Implemented with thermal operations:} It consists of a thermal
operation that operates on $X^n$ and a battery $W$ (along with an ancillary heat
bath that appears in the thermal operation);
(ii)~\emph{Optimal variable work cost:} When given the input state
$\sigma^{\otimes n}$, it consumes an amount of work per copy asymptotically
equal to
$\beta^{-1}`\big[ \DD{\mathcal{E}(\sigma)}{\Gamma} - \DD{\sigma}{\Gamma}]$;
(iii)~\emph{Universal:} It universally implements the map
$\mathcal{E}^{\otimes n}$, up to a dephasing operation that amounts to revealing
to the environment how much work was used.  We show that for any
time-covariant $\sigma_X$, the implementation outputs
$[\mathcal{E}(\sigma_X)]^{\otimes n}$ when given as input
$\sigma_X^{\otimes n}$, while a specific type of dephasing affects the
correlations between $X$ and an input purifying reference system.
The case of work extraction is recovered by considering the process
$\mathcal{E}(\cdot) = \tr`(\cdot)\,\gamma$.

A core technical contribution involves adapting our conditional erasure protocol
in Ref.~\cite{Faist2021CMP_impl}.  \emph{Conditional erasure} is the task of
resetting a system $S$ to some standard state while having access to some memory
register $M$~\cite{delRio2011Nature,Faist2015NatComm}.  (This task can be
thought of as \emph{conditional thermalization} or even \emph{conditional work
  extraction}, if the final fixed state of $S$ is the thermal state.)
Ref.~\cite{Faist2021CMP_impl} provides the tools necessary to create a
\emph{semiuniversal} implementation of conditional erasure at a deterministic
work cost of some predetermined amount $w$, and which universally implements
conditional erasure over the set of all states whose individual erasure work
cost does not exceed the threshold $w$.  A universal, variable-work-cost
conditional erasure protocol can be obtained by performing a suitable
measurement on $SM$ and conditionally applying a semiuniversal protocol with the
necessary threshold value $w$ depending on the measurement outcome.
The conditional erasure protocol translates to an implementation of
$\mathcal{E}$ as follows: We first implement the unitary Stinespring dilation of
$\mathcal{E}^{\otimes n}$ at no work cost involving some ancillary systems $E^n$
acting as the environment, and then proceed to conditionally erase $E^n$ using
$X^n$ as our memory register.

We further investigate a decoherence effect that is inherent to a variable-work
implementation of a quantum process.  Any such process reveals the amount of
work consumed in the final charge state of the battery.  Yet this amount of
consumed work must be correlated with the states involved in the process given
that the process uses a variable amount of work.  Also, the implementation
interacts differently with the environment when different amounts of work are
consumed.  Revealing this information causes, in general, a form of decoherence
of the process that can significantly affect the final state.  Yet, for any
time-covariant i.i.d.\@ input state, this decoherence leaves the output state on
the system invariant and only affects its correlation with a reference system.

In the following, we introduce the thermodynamic framework and notation
(\cref{sec:Setup}) and revisit the simpler situation of work extraction from a
state (\cref{sec:WarmupWorkExtraction}).  We then present our main results:
first for conditional erasure (\cref{sec:ConditionalErasureVariable}), followed
by our general optimal, universal, variable-work-cost thermodynamic protocol for
time-covariant processes (\cref{sec:MainResult}).  We proceed to discuss the
effect of the work-cost dephasing (\cref{sec:WorkCostCovariance}).  We conclude
in \cref{sec:Outlook}.

\section{Thermodynamic model and notation}
\label{sec:Setup}

We follow closely the definitions and notation used in our earlier
Ref.~\cite{Faist2021CMP_impl}.  We focus on the resource theory of thermal
operations (Refs.~\cite{%
  BookBinder2018_ThermoQuantumRegime,Goold2016JPA_review,%
  Brandao2013_resource,Horodecki2013_ThermoMaj,Brandao2015PNAS_secondlaws,%
  Faist2015NatComm,PhDPhF2016} and references therein).
We consider a quantum system $X$ of finite dimension, with any arbitrary
Hermitian operator $H_X$ as Hamiltonian.

We imagine an agent whose goal is to implement a particular fixed quantum
evolution, specified as a completely positive, trace-preserving map
$\mathcal{E}_{X\to X}$.
We assume that the process is \emph{time-covariant}, meaning that
$\mathcal{E}_{X\to X}`[ \ee^{-iH_Xt} \, `(\cdot)\, \ee^{iH_Xt} ] = \ee^{-iH_Xt}
\, \mathcal{E}_{X\to X}`(\cdot) \, \ee^{iH_Xt}$ for all $t$; this assumption
sidesteps issues of defining a work cost when manipulating
coherence~\cite{Faist2015NJP_Gibbs,%
  Lostaglio2015NC_beyond,Lostaglio2015PRX_coherence}.

We suppose that the agent operates with an environment at fixed inverse
temperature $\beta>0$; concretely, they can: (i)~bring in any finite-dimensional
ancillary system $B$, with any arbitrary Hamiltonian $H_B$, initialized in its
thermal state $\gamma_B = \ee^{-\beta H_B} / \tr`\big(\ee^{-\beta H_B})$;
(ii)~apply any global energy-conserving unitary $U_{XB}$ acting on $X$ and any
relevant heat baths included from point (i), where energy conservation refers to
commuting with the total Hamiltonian, $U_{XB} `(H_X+H_B) = `(H_X+H_B) U_{XB}$;
and (iii)~trace out any subsystem.
The agent can implement a map $\mathcal{E}_{X\to X}$ to arbitrary precision if
and only if $\mathcal{E}_{X\to X}$ is a \emph{thermal operation}, which can be
written as a sequence consisting of an operation of type~(i) followed by one of
type~(ii) and a final tracing out of $B$, which is an operation of type~(iii).

We define for later convenience some useful quantities: For any system $Y$ (in
particular $Y=X,B$), let $\Gamma_Y = \ee^{-\beta H_Y}$ and
$Z_Y = \tr(\Gamma_Y)$.  The thermal state of $Y$ is $\gamma_Y = \Gamma_Y/Z_Y$,
and its equilibrium free energy is $F_Y = -\beta^{-1}\log(Z_Y)$.

To implement a map $\mathcal{E}_{X\to X}$ that is not a thermal operation, we
need to invest \emph{thermodynamic work}.  We imagine a \emph{battery system}
$W$ whose state we ensure remains among a family of \emph{battery states}
$\tau^{(E)}_W$.
A thermal operation that operates on a system $S$ and a battery $W$ (using a
bath $B$) can effect transitions on $W$ between different states of the form
$\tau^{(E)}_W$, which is then interpreted as charging or depleting the battery
(\cref{fig:TOWork}).
\begin{figure}
  \centering
  \includegraphics{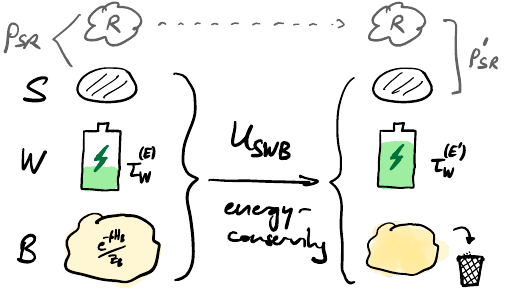}
  \caption{Thermal operations with a battery can be used to implement an
    operation on the system $S$ while consuming work; they involve a global
    energy-conserving unitary $U_{SWB}$ operating on $S$, a battery $W$ with a
    family of charge states $`{\tau_W^{(E)}}_E$, and a heat bath $B$.  The
    quantum channel induced on $S$ by the thermal operation is revealed in the
    correlations between $S$ and $R$ in the output state $\rho'_{SR}$, for
    entangled inputs $\rho_{SR}$.  The battery's output charge state is
    guaranteed to be at least $E'$ if its state passes a hypothetical test
    represented by a POVM effect $\Pi_W^{\geq E'}$; the battery's output state
    can then also be assumed uncorrelated with the system by explicitly
    thermalizing it within $\Pi_W^{\geq E'}$'s support.  The bath is eventually
    traced out.  Multiple heat baths can be combined into a single large heat
    bath.  Similarly, multiple batteries can equivalently be combined into a
    single one, because the battery states are reversibly transformable with
    thermal operations; therefore, a sequence of thermal operations operating on
    different batteries is equivalent to a single operation acting on a single
    battery using the sum of the work costs.  Pure time-covariant ancillary
    systems can be borrowed at the cost of transforming a suitable bath system
    from its thermal state to the desired state; such an ancilla may be used,
    for instance, to perform for free another thermal operation conditionally on
    a measurement outcome.}
  \label{fig:TOWork}
\end{figure}
Many different models of battery states are
equivalent~\cite{Horodecki2013_ThermoMaj,Brandao2015PNAS_secondlaws,%
  Faist2015NatComm,PhDPhF2016}, and we choose the following family for
convenience.  We ensure $W$'s dimension is sufficiently large with the trivial
Hamiltonian $H_W = 0$; we take $\tau_W^{(E)}$ to be \emph{information battery
  states}, which have uniform spectrum over a rank
$r^{(E)} := e^{\beta(E - E_{\mathrm{batt.ref.}})}$:
\begin{align}
  \tau_W^{(E)}
  := \frac1{r^{(E)}} \sum_{j=1}^{r^{(E)}} \proj{j}_W\ ,
  \label{eq:InformationBatteryStates}
\end{align}
where we choose $E_{\mathrm{batt.ref.}}$ sufficiently negative that there is a
battery state $\tau_W^{(E)}$ for a sufficiently dense set of values of $E$ to
approximate the energy values relevant to our problem.
The states $\tau_W^{(E)}$ have the following property: If $A$ is any system with
at least two energy levels $0,E$, then the transformations
$\proj{0}_A\otimes\tau^{(E)} \to \ket{E}_A\otimes\tau^{(0)}$ and
$\proj{E}_A\otimes\tau^{(0)} \to \ket{0}_A\otimes\tau^{(E)}$ are both possible
with thermal operations, to arbitrary accuracy as
$E_{\mathrm{batt.ref.}}\to-\infty$, $\abs{B}\to\infty$.  For this reason, the
battery $W$ occupying the state $\tau^{(E)}_W$ can be thought of as carrying an
energy charge $E$.
We can \emph{test the battery charge} using the following projector.
Let $\Pi_W^{\geq E}$ be the following rank-$e^{\beta(E-E_{\mathrm{batt.ref.}})}$
projector:
\begin{align}
  \Pi_W^{\geq E}
  := \sum_{j=1}^{r^{(E)}} \proj{j}_W\ .
\end{align}
Then $\tr`\big[ \tau^{(E')}_W \Pi_W^{\geq E} ] = 1$ if and only if $E' \geq E$.
If, after a measurement $`\big{ \Pi_W^{\geq E}, \Ident_W - \Pi_W^{\geq E} }$ of
the battery, we get the outcome $\Pi_W^{\geq E}$, then the battery has charge at
least $E$ (relative to the reference level $E_{\mathrm{batt.ref.}}$).  In this
case, we can deterministically ensure $W$ contains exactly $\tau_W^{(E)}$ by
thermalizing $W$ on $\Pi_W^{\geq E}$'s support using a free operation.

The work cost of converting $\sigma_X\to \rho_X$, when these states both commute
with $H_X$, can be computed using the notion of
\emph{thermomajorization}~\cite{Horodecki2013_ThermoMaj}.  A case if interest is
when $\sigma_X=\proj{E}_X$ is a pure energy eigenstate and when
$\rho_X = \gamma_X$, is the thermal state of $X$, with
$\gamma_X = \ee^{-\beta H_X}/Z_X$ and $Z_X = \tr`(\ee^{-\beta H_X})$; then the
amount of work $E + \beta^{-1}\log Z_X$ can be extracted in this transformation
(to arbitrary precision for a large battery).  Conversely, the same amount of
work needs to be expended to carry out the reverse operation,
$\gamma_X \to \ket{E}_X$.  Energy eigenstates are therefore mutually reversibly
interconvertible and reversibly interconvertible to the thermal state.  (This
property is shared more generally by all states of the form
$\rho \propto \Pi \gamma_X \Pi$, where $\Pi$ is a projector that commutes with
$\gamma_X$.)

Let $\sigma_{X}$ be any quantum state and let $\epsilon\geq 0$.  An
\emph{implementation of $\mathcal{E}^{\otimes n}$ for $\sigma_X$} is a thermal
operation acting on $X^n$ and a battery $W$ initialized in some suitable battery
state $\tau_W^{(E_0)}$, such that when given as input
$\ket\sigma_{XR}^{\otimes n}$, the implementation produces a state
$\epsilon$-close in fidelity to
$\mathcal{E}^{\otimes n}(\sigma_{XR}^{\otimes n})$.
(See \cref{fig:ImplProcUniv}.)
\begin{figure}
  \centering
  \includegraphics{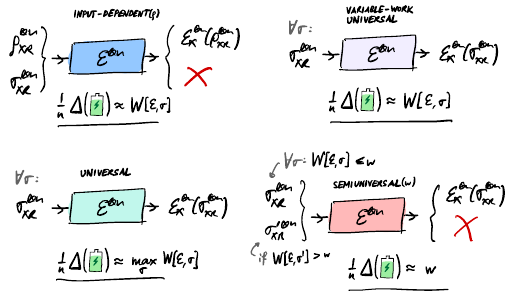}
  \caption{Types of thermodynamic implementations of an i.i.d.\@ quantum process
    $\mathcal{E}$.  An implementation tailored for $\rho_{XR}$ fails if any
    other state is used as an input.  A \emph{universal} implementation produces
    the correct output for all input states, using a deterministic amount of
    work per copy.  A \emph{variable-work universal} implementation produces the
    correct output (up to a dephasing) for all input states, using an amount of
    work that varies depending on the actual input.  A \emph{semiuniversal}
    implementation for work $w$ behaves like a universal implementation, but it
    only produces the correct output if the work required to implement the
    process for the given input state does not exceed $w$.}
  \label{fig:ImplProcUniv}
\end{figure}
The work cost $w$ of the implementation is determined by how much charge is left
in the battery $W$, which can be tested by ensuring $W$'s new state passes the
measurement $\Pi_W^{(\geq E_0-w)}$.
Asymptotically for large $n$, the optimal deterministic work cost of an
implementation of $\mathcal{E}^{\otimes n}$ for $\sigma_X$ is an amount
$W[\mathcal{E}_X;\sigma_X]$ of work per copy given by~\cite{Faist2015NatComm,%
  Faist2018PRX_workcost}
\begin{align}
  W[\mathcal{E}_X;\sigma_X]
  = \beta^{-1} `\big[
  \DD{\mathcal{E}_X`(\sigma_X)}{ \gamma_X }
  - \DD{\sigma_X}{\gamma_X} ]\,.
\end{align}

We may instead demand the implementation perform accurately for multiple
possible input states.
Let $\mathcal{S}_{X^n}$ be a set of states on $X^n$.  An implementation of
$\mathcal{E}^{\otimes n}$ is \emph{universal with respect to
  $\mathcal{S}_{X^n}$} if, for any $\ket\sigma_{X^nR}$ with $\sigma_{X^n}$ in
the convex hull of $\mathcal{S}_{X^n}$, the implementation approximates the
output state $\mathcal{E}^{\otimes n}(\sigma_{X^nR})$.

When $\mathcal{S}_{X^n}$ is the set of all i.i.d.\@ input states (and if the
implementation is permutation-invariant over the $n$ copies of the system), the
condition becomes equivalent to demanding proximity in diamond norm between the
implementation's effective channel and the target channel
$\mathcal{E}^{\otimes n}_{X\to X}$, up to a $\poly(n)$ factor in the error
parameter~\cite{Christandl2009PRL_Postselection,%
  Faist2019PRL_thcapacity,Faist2021CMP_impl}.  Then the implementation is
universal with respect to the set of all quantum states on $X^n$.
In this case, the optimal deterministic work cost counted per copy in the
asymptotic limit $n\to\infty$ is the \emph{thermodynamic
  capacity}~\cite{Faist2019PRL_thcapacity,Faist2021CMP_impl,%
  Gour2021PRR_entropychannel},
\begin{align}
  T(\mathcal{E})
  &= \max_{\sigma}  \ W[\mathcal{E};\sigma]
  \nonumber\\
  &= \max_{\sigma}\,
  \beta^{-1} `\big[ \DD{\mathcal{E}(\sigma_X)}{ \ee^{-\beta H_X} }
    - \DD{\sigma_X}{\ee^{-\beta H_X}}
  ]
    \, .
\end{align}

An intermediate situation, somewhere between an implementation being specific to
a state $\sigma_X$ and it being universal with respect to all input states, will
appear in our proof.  
More precisely, we call a \emph{semiuniversal implementation of a process at
  work cost $w$} one for which the work cost is deterministically $w$ per copy,
asymptotically for $n\to\infty$, and which is universal with respect to the set
of all input states $\sigma^{\otimes n}$ for which $W[\mathcal{E};\sigma]$ does
not exceed $w$.
To shape the reader's intuition, a similar situation appears for data
compression of a quantum state.  A state $\rho^{\otimes n}$ can be compressed
onto $nS`(\rho) + O`(\sqrt{n})$ qubits, as known from quantum Shannon theory.
There cannot be any useful universal compression protocol for arbitrary
$\rho^{\otimes n}$, since such a protocol must accommodate the maximally mixed
state which cannot be compressed.  Yet if we are promised that the state's
entropy does not exceed some threshold, $S`(\rho)\leq m$, then we can use a
single compression protocol to compress any state below this threshold onto
$nm+O`(\sqrt{n})$ qubits, using universal variants of typical
projectors~\cite{Jozsa1998PRL81,Jozsa2003}.  Such a process is a data
compression analog to the semiuniversal processes considered here.

We now formalize the concept of a \emph{variable-work} implementation of
$\mathcal{E}^{\otimes n}$, by extending the concept of deterministic work usage
from the battery $W$.
A thermal operation operating with a battery $W$ initialized in $\tau_W^{(E_0)}$
is said to use an \emph{$\epsilon$-deterministic amount of work $w$} if its
output state on $W$ has overlap at least $1-\epsilon$ with
$\Pi_W^{\geq(E_0 - w)}$.  The protocol can be seen as ``depleting'' the battery
by an amount $w$.  If we hypothetically were to test the battery's final charge
to be at least $E_0-w$, our test would be guaranteed to succeed with high
probability.  The final state on $W$ can be forced to be equal to
$\tau_W^{(E_0-w)}$ by applying a free thermalization operation on the support of
$\Pi_W^{\geq(E_0-w)}$.
In contrast, a process that uses a \emph{variable amount of work} is allowed to
deplete the battery at different rates depending on the input state.  In such
cases, the implementation guarantees that the test $\Pi_W^{\geq(E_0-w_\sigma)}$
passes with high probability, where $w_\sigma$ may now depend on the input state
$\sigma$.
In all cases, this test remains hypothetical and we never actually apply this
projector in the implementation itself.  Rather, the implementation guarantees
that if the user decided to test the battery's final charge, the test would pass
with high probability.  This guarantee is sufficient for any subsequent protocol
using the same battery to assume that its charge state is at least
$E_0-w_\sigma$, up to an error probability $\epsilon$; indeed, a subsequent
protocol's failure because of insufficient battery charge can be viewed as a
test of the battery charge; such a test cannot find that the battery's charge is
less than $E_0-w$ any better than the measurement $\Pi_W^{\geq(E_0-w)}$.

The above framework has the following useful features.  First, a sequence of
thermal operations is again a thermal operation, because the bath systems can be
combined into one large bath and the separate unitaries merged into one large
energy-conserving unitary.
Similarly, multiple batteries can be combined into a single sufficiently large
battery; the total battery charge difference is the sum of the individual
battery charge differences.
Then we can carry out time-covariant projective measurements and subsequent
conditional energy-conserving unitaries by borrowing a thermal ancilla,
resetting it to a pure state using work, performing an energy-conserving
interaction to carry out the desired measurement.

We assume our battery system to be sufficiently large that it can provide or
store any amount of energy within the range of possible work requirements of the
process, to high accuracy.  Other, more rudimentary battery models can only be
charged by one specific amount of energy fixed in advance.  Such is the case,
for instance, of a \emph{work bit} used in Ref.~\cite{Horodecki2013_ThermoMaj}
consisting of two energy levels.  Such batteries suffice for a scenario like a
state transformation, where the amount of work is deterministic and can be
computed in advance.  However, in a variable-work implementation of a process,
the work requirement is by definition not known in advance; we provision a
sufficiently large battery to accommodate any amount of work requirement that
can arise in the process.

While we choose an information battery with
states~\eqref{eq:InformationBatteryStates} for our implementation, any other
standard type of battery that can accommodate the necessary possible charge
differences can be used instead.  An alternative model, for instance, is a
system with evenly-spaced energy levels covering a large range of energy values,
in the regime of very small energy spacing.  The battery states are the energy
eigenstates, and the energy labels the battery charge.  This type of battery is
often referred to as a ``weight,'' akin to a physical weight one can raise or
lower to store or consume energy.  To construct a protocol with this type of
battery, we bring in a sufficiently large thermal heat bath ancilla with a
trivial Hamiltonian (to be used as an information battery), use a partial
initial charge from the weight battery to charge the information battery, run
the process with the information battery, and restore any remaining information
battery charge to the weight energy battery.  While some initial amount of
energy charge in the energy battery might be required, this can typically be
attributed to an arbitrary setting of the zero reference value for energy,
especially in the presence of the sophisticated control systems that would be
necessary to implement these operations in the first place.

\section{Warm-up: Universal work extraction with optimal per-input work cost}
\label{sec:WarmupWorkExtraction}

In this section, we illustrate how a semiuniversal protocol for work extraction
can be turned into an variable work cost protocol with optimal per-input work
cost, and explain how results in Ref.~\cite{arXiv:1911.05563} can be pieced
together to obtain this result.

First of all, observe that work extraction corresponds to implementing the
completely thermalizing map $X\to X$, i.e.,
$\mathcal{E}_{X\to X}(\cdot) = \tr_X(\cdot)\,\gamma_X$ with
$\gamma_X = \ee^{-\beta H_X}/\tr`(\ee^{-\beta H_X})$.  Indeed, any work
extraction protocol can be followed by a free thermalization map to reach
$\gamma_X$ and therefore implement $\mathcal{E}$ accurately.  (We discuss this
point in deeper detail below, in \cref{sec:ConditionalErasureVariable}.)

If the protocol is tailored to a specific i.i.d.\@ input state
$\sigma_X^{\otimes n}$, then the optimal amount of work that can be extracted
per copy is, asymptotically,
$F_X(\sigma_X) := \beta^{-1}\DD{\sigma_X}{\ee^{-\beta H_X}}$%
~\cite{Horodecki2013_ThermoMaj}.  We seek to construct, as in
Ref.~\cite{Watanabe2025arXiv_universal}, a \emph{universal} protocol that always
achieves the variable extracted amount of work $F_X(\sigma_X)$ per copy, for
each individual $\sigma_X^{\otimes n}$.

We first consider a \emph{semiuniversal} protocol for work extraction.  Namely,
given some value $w$, we design a protocol with the property that it always
deterministically extracts $w$ work, but subject to the promise that the input
state obeys $F_X(\sigma_X) \geq w$.  The protocol can behave arbitrarily (it may
even mess up the battery state) if it is applied onto a state
$\sigma_X^{\otimes n}$ that fails this input condition.  Such a protocol is
immediately given by \cref{thm:UniversalTheoremCovariantInputSet} below in
\cref{appx:UniversalWithSetOfInputs} (use
$\mathcal{E}_{X\to X}(\cdot)=\tr_X`(\cdot)\,\gamma_X$), which is itself a
straightforward adaptation of Theorem~7.1 of Ref.~\cite{arXiv:1911.05563}.  More
specifically, such a semiuniversal protocol is given for example by
Proposition~7.1 of Ref.~\cite{arXiv:1911.05563}, using $S = X^n$ and with $M$
being a trivial system which we ignore entirely, and taking
$P_{SM} \equiv P_{X^n}$ to be the product of a Schur-Weyl measurement (which
estimates the input state's entropy) and a global energy measurement (which
estimates the input state's energy), summed over outcomes which imply the input
free energy is at least $w$.  Alternatively, this projector is directly given by
Proposition~6.1 of Ref.~\cite{arXiv:1911.05563} where we take $B$ to be a
trivial system.

We construct a universal, optimal variable-work-cost protocol for work
extraction as follows.  First, estimate the input state's free energy, by
combining a Schur-Weyl block measurement along with a global energy measurement.
Then, apply the corresponding semiuniversal work extraction protocol for the
obtained free energy estimate.  (More specifically, if the free energy estimate
is $w'$, we apply the semiuniversal protocol for $w'+\delta$ for arbitrarily
small but constant $\delta$, to ensure that the failure probability remains
exponentially suppressed in $n$.  At the end we take $\delta\to0$.)  The
measurement can be integrated within a global thermal operation describing the
entire process, leaving a memory register storing the estimated free energy
value.  Because there are at most $\poly(n)$ Schur-Weyl blocks and energy
outcomes, resetting this memory contributes a vanishing asymptotic work cost per
copy.  This protocol achieves the desired behavior at the requested work cost.

The proof of our main results (see below) follow a very similar strategy.

\section{Conditional thermodynamic erasure with variable work;
  equivalence with work extraction}
\label{sec:ConditionalErasureVariable}

We first turn our attention to the task of resetting a quantum system when we
have access to a quantum memory~\cite{delRio2011Nature,Faist2015NatComm,%
  Faist2018PRX_workcost,Faist2021CMP_impl} (\cref{fig:CondReset}).
\begin{figure}
  \centering
  \includegraphics{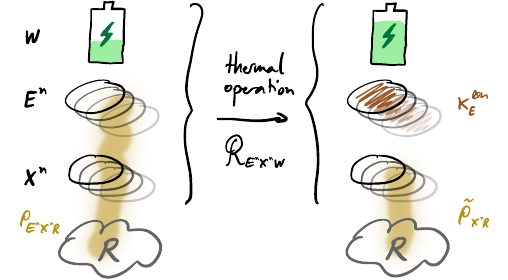}
  \caption{Implementation of a universal, variable-work conditional reset
    operation of systems $E^n$ conditioned on $X^n$ (Main Result; version for
    conditional reset).  The systems $E^n$ and $X^n$ start in some joint pure
    state $\ket\rho_{E^nX^nR}$ with a reference system $R$.  The reset operation
    should reset each copy of $E$ to the fixed state $\kappa_E$ while preserving
    the reduced state $\rho_{X^nR}$.  Here we suppose $\kappa_E$ is the thermal
    state on $E$; a similar argument holds if $\kappa_E$ is any energy
    eigenstate.  (If $X^n$ is trivial, this task is equivalent to work
    extraction.)  Our variable-work implementation $\mathcal{R}_{E^nX^nW}$
    consumes an amount of work $w$, using a battery $W$, which is asymptotically
    optimal for each i.i.d.\@ input state $\rho_{EX}^{\otimes n}$.  The output
    state $\tilde\rho_{X^nR}$ on $X^nR$ differs from $\rho_{X^nR}$ by a
    dephasing operation due to the fact that the process interacts differently
    with the environment for different amounts of consumed work (which is
    revealed in the battery's output state).  If $\rho_{E^nX^n}$ is
    time-covariant and permutation-invariant, we can show that, in fact,
    $\tilde\rho_{X^n} = \rho_{X^n}$.}
  \label{fig:CondReset}
\end{figure}
Consider two quantum systems $E$ and $X$ in a state $\rho_{EX}$.  The state is
purified in a reference system $R$ as $\ket\rho_{EXR}$.  We fix a reference
state $\kappa_E$.  Our task is to reset the $E$ system to the $\kappa_E$ state,
using the least amount of thermodynamic work necessary, while ensuring that the
reduced state on $\rho_{XR}$ remains unchanged.  We may act on systems $E$ and
$X$, but we are not allowed to touch the $R$ system.  Specifically, we must find
a thermal operation acting on $EX$ along with a battery $W$ such that, when
given the input state $\rho_{EXR}$, the output is $\epsilon$-close to
$\kappa_E\otimes\rho_{XR}$.
The reason we choose the system names $E$, $X$ will become apparent when we use
our results in the next section to implement a general time-covariant process.

The presence of the $X$ system can help us achieve a lower work cost that if we
could only act on $E$, since $X$ might carry precious information about $E$'s
state through correlations between $X$ and $E$~\cite{delRio2011Nature}.  For
example, consider the state $\rho_{EX} = [\proj{00}_{EX}+\proj{11}_{EX}]/2$ on
two qubits $E$ and $X$ with trivial Hamiltonians $H_E=H_X=0$.  Here, $E$ can be
conditionally reset to $\ket0_E$ at no work cost, by applying a CNOT operation
between $X$ (the control) and $E$ (the target), and dephasing $X$ in the
computational basis.  Note this protocol also preserves the reduced state
$\rho_{XR} = [\proj{00}_{XR} + \proj{11}_{XR}]/2$, using the purification
$\ket\rho_{EXR}=[\ket{000}+\ket{111}]/\sqrt2$.  Yet a protocol operating on $E$
alone would have to map $\Ident_E/2$ to $\ket0_E$, which requires
$\beta^{-1}\log(2)$ work.
 
Multiple tasks can be formulated in terms of this setting of conditional erasure
to $\kappa_E$:
\begin{itemize}
\item If $\kappa_E = \proj{0}_E$, where $\ket{0}_E$ is some pure energy
  eigenstate, conventionally of zero energy, the task corresponds to a
  conditional erasure to $\ket0_E$.  This task is naturally referred to as
  \emph{conditional erasure} or \emph{conditional
    reset}~\cite{delRio2011Nature}.
  In the case of a trivial memory system $X$, this is the setting of Landauer's
  original \emph{erasure} task~\cite{Landauer1961_5392446Erasure}.
\item If $\kappa_E = \gamma_E$, then the task is to thermalize $E$ while
  extracting as much work as possible.  Here, correlations with $X$ may be
  exploited for work gain.  This task can be thought of as a \emph{conditional
    thermalization}.  For a trivial memory system $X$, this task is simply a
  \emph{thermalization of $E$}: it replaces $E$'s state with its thermal state.
\item The task of \emph{work extraction} from a state $\rho_E$ is to extract as
  much work as possible from $\rho_E$, with no regard as to what the final state
  on $E$ is.  Without loss of generality, the final state is $\gamma_E$, since
  any other state can be thermalized to $\gamma_E$ with a free thermal
  operation.  Therefore, the task of \emph{work extraction} naturally implements
  a \emph{thermalization of $E$} (with trivial $X$).  Conversely, any
  thermalization of $E$ can be viewed as a work extraction process, yielding
  some amount of work from $\rho_E$.  Therefore, the tasks of
  \emph{thermalization of $E$} and \emph{work extraction from $\rho_E$} are
  equivalent.
\end{itemize}

All the tasks above can be reduced to the convenient case where
$\kappa_E = \gamma_E$.  This holds because $\gamma_E$ is reversibly
interconvertible to any energy eigenstate
$\ket{E}_E$~\cite{Brandao2013_resource,Horodecki2013_ThermoMaj,%
  Aberg2013_worklike,Brandao2015PNAS_secondlaws}: The deterministic,
single-instance work cost of $\gamma_E\to\ket{E}_E$ is $nF_E- E$; whereas the
reverse process can yield all this work back, at a work cost $E - nF_E$.
So any optimal protocol for conditional erasure to $\ket{E}_E$ automatically
yields an optimal protocol for conditional erasure to $\gamma_E$, because we can
append a reversible protocol to convert $\ket{E}_E \to \gamma_E$; and the same
holds conversely if we are given an optimal protocol for conditional erasure to
$\gamma_E$.
Henceforth, we set $\kappa_E = \gamma_E$.

The task of conditional erasure to $\gamma_E$ for $\ket\rho_{EXR}$ means
implementing the map
\begin{align}
  \mathcal{E}_{\mathrm{cond.reset}\,EX}(\cdot) = \gamma_E \tr_E(\cdot) \ ,
\end{align}
for the input state $\rho_{EX}$.  A \emph{universal} implementation of
$\mathcal{E}_{\mathrm{cond.reset}\,EX}^{\otimes n}$ for a set of states
$\mathscr{S}_{E^nX^n}$ ensures that for any $\ket\rho_{E^nX^nR}$ with
$\rho_{E^nX^n}$ in the convex hull of $\mathscr{S}_{E^nX^n}$, the state of $E$
is reset to $\gamma_E^{\otimes n}$ while preserving the reduced state
$\rho_{X^nR}$.

Here, we propose a universal, variable-work-cost implementation of conditional
erasure to $\gamma_E$ that uses an optimal amount of work for each i.i.d.\@
input state.

The idea for our conditional erasure implementation is that the process first
carries out a global measurement that estimates how much work $w$ we anticipate
requiring to implement the conditional erasure.  Depending on the outcome $w$,
we apply a corresponding semiuniversal protocol, namely one that can
successfully erase any state $\rho_{EX}$ whose asymptotic erasure work cost
$W[\tr_E;\rho_{EX}]$ is up to at most $w$.

To construct this measurement more precisely, we need some technical
definitions.  We follow closely the notation and conventions used in
Ref.~\cite{arXiv:1911.05563}, cf.\@ also
Refs.~\cite{PhDHarrow2005,Haah2017IEEETIT_sampleoptimal}.
Let $\Pi_{X^n}^\lambda$ be the projector onto the Schur-Weyl block of $X^n$
labeled by the Young diagram $\lambda$.
If we normalize $\lambda$ into a
probability distribution, $\bar\lambda := \lambda/n$, we denote its Shannon
entropy as $S(\bar\lambda) = -\sum (\lambda_i/n)\log(\lambda_i/n)$.
Let $`{ S_{X^n}^{E'_\ell} }$ and $`\big{ R_{E^nX^n}^{E_k} }$ be the projectors
onto the global eigenspaces of $H_{X^n} := \sum_{i=1}^n H_{X_i}$ and
$H_{E^nX^n} := \sum_{i=1}^n (H_{E_i} + H_{X_i})$, where $E_k$ and $E'_\ell$
label the corresponding eigenvalues of $H_{X^n}$ and $H_{E^nX^n}$, respectively.
Each $R_{E^nX^n}^{E_k}$ and $S_{X^n}^{E'_\ell}$ commutes with $\Pi_{X^n}^\mu$
and $\Pi_{(EX)^n}^\lambda$ as it is permutation-invariant.  Each set
$`{ R_{E^nX^n}^{E_k} }$, $`{ S_{X^n}^{E'_\ell} }$, $`{ \Pi_{(EX)^n}^\lambda }$,
$`{ \Pi_{X^n}^\mu }$ forms a POVM with at most $\poly(n)$ outcomes; any pairs of
effects from these sets commute and furthermore commute with $H_{X^n}$ and
$H_{E^nX^n}$.

Let
\begin{align}
  P^{(w)}_{E^nX^n}
  := 
  \hspace*{-1.5em}
  \sum_{\substack{k,\ell,\lambda,\mu\ :\\
  E_k - \beta^{-1}S(\lambda/n) - E_\ell + \beta^{-1}S(\mu/n) \;=\; w}}
  \hspace*{-3em}
  S_{X^n}^{E_\ell} \Pi^{\mu}_{X^n} \Pi^{\lambda}_{(EX)^n} R_{E^nX^n}^{E_k}\ ,
  \label{eq:WorkCostMeasEnvXpP}
\end{align}
where $w$ ranges over a $\poly(n)$-sized set $\mathscr{W}$ that contains all
possible real values that can be attained by different choices of
$k,\ell,\lambda,\mu$.  The operators $`{ P^{(w)}_{E^nX^n} }_{w\in\mathscr{W}}$
form a POVM.

Measuring the POVM $`{ P^{(w)}_{E^nX^n} }_{w\in\mathscr{W}}$ induces a
decoherence in the global $n$-copy systems along subspaces with different
``information about how much work is required for erasure.''  This dephasing
process is:
\begin{align}
  \mathcal{D}^{\textup{W}}_{E^nX^n}(\cdot)
  = \sum_{w\in\mathscr{W}} P_{E^nX^n}^{(w)}\,`(\cdot)\, P_{E^nX^n}^{(w)}\ .
  \label{eq:WorkCostDephasingMapEX}
\end{align}
Therefore, the variable-work-cost implementation presented here solves the task
of conditional erasure, \emph{up to the dephasing $\mathcal{D}^{\textup{W}}$}
that corresponds to knowing how much work was actually expended in the process.
We discuss the meaning of this dephasing operation in the following sections.

\begin{theorem}[\textbf{Main Result; Conditional Erasure}]
  \label{thm:UnivCondErasureForEachIidInput}
  Let $E,X$ be quantum systems with Hamiltonians $H_E$, $H_X$, and let
  $\Gamma_X = e^{-\beta H_X}$, $\Gamma_E = e^{-\beta H_E}$,
  $\Gamma_{XE} = \Gamma_X\otimes\Gamma_E$ and $\gamma_E=e^{-\beta H_E}$.  Let
  $\delta>0$.  Then there exists a thermal operation $\mathcal{R}_{E^nX^nW}$
  acting on a battery $W$ and there exists $\eta>0$ such that:
  \begin{enumerate}[label=(\roman*)]
  \item\label{item:thmUnivCondErasureForEachIidInputOutputFidelity}
    For any arbitrary state $\rho_{E^nX^n}$, with any purification
    $\ket\rho_{E^nX^nR}$,
    \begin{multline}
      F`\Big( 
      \tr_W`\big[ \mathcal{R}_{E^nX^nW}`\big( \rho_{E^nX^nR} \otimes \tau^{(E_0)}_W ) ]
      ,
      \gamma_E^{\otimes n} \otimes \tilde\rho_{X^nR}
      )
      \\[1ex]
      \geq 1 - \poly(n)\,e^{-n\eta}\ ,
      \label{eq:UnivCondErasureForEachIidInputCondCorrect}
    \end{multline}
    where
    $\tilde\rho_{E^nX^nR} = \mathcal{D}^{\textup{W}}_{E^nX^n}[\rho_{E^nX^nR}]$
    is the state $\rho_{E^nX^nR}$ that underwent a dephasing in the subspaces
    that correspond to different work cost value requirement.

    Furthermore, for any time-covariant $\rho_{EX}$, we have
    $\mathcal{D}^{\textup{W}}_{E^nX^n}[\rho_{EX}^{\otimes n}] = \rho_X^{\otimes
      n}$;
  \item \label{item:thmUnivCondErasureForEachIidInputWorkCostIidInput}
    For any state $\rho_{EX}$, we have
    \begin{multline}
      \tr`\Big[ \Pi_W^{\geq (E_0 - w_\rho)}
      \mathcal{R}_{E^nX^nW}`\big( \rho_{EX}^{\otimes n} \otimes \tau^{(E_0)}_W ) ]
      \\[1ex]
      \geq 1 - \poly(n)\,e^{-n\eta/2}\ ,
      \label{eq:UnivCondErasureForEachIidInputCondWCost}
    \end{multline}
    where
    \begin{align}
      w_\rho &=
      n\bigl[ \beta^{-1} D(\rho_X\Vert\Gamma_X)
      - \beta^{-1}D(\rho_{EX}\Vert\Gamma_{EX}) 
               \nonumber\\[1ex]
      &\qquad\ + F_E + 7\beta^{-1}\delta\bigr] + \log\poly(n)
      \ .
      \label{eq:UnivCondErasureForEachIidInputWCostValue}
    \end{align}
  \item \label{item:thmUnivCondErasureForEachIidInputWorkCostAnyInput}
    For any arbitrary state $\rho_{E^nX^n}$, and for any $w'$, we have
    \begin{multline}
      \tr`\Big[ \Pi_W^{\geq (E_0-w'')}
      \mathcal{R}_{E^nX^nW}`\big( \rho_{E^nX^n} \otimes \tau^{(E_0)}_W ) ]
      \\[1ex]
      \geq \sum_{w\leq w'} \tr`\Big[ P^{(w)}_{E^nX^n} \rho_{E^nX^n} ]
      - \poly(n)\,\ee^{-n\eta/2}\ ,
      \label{eq:UnivCondErasureForEachIidInputCondWChargeArbitraryInput}
    \end{multline}
    where
    $
      w'' = n`\big(w' + 6\beta^{-1}\delta + F_E) + \log\poly(n)
    $.
  \end{enumerate}
\end{theorem}

The proof of \cref{thm:UnivCondErasureForEachIidInput} appears in
\cref{appx:UnivImplPerIidInputCondErasureProofs}.

Condition~\eqref{eq:UnivCondErasureForEachIidInputCondWCost} ensures that the
final battery charge is at least $E_0 - w_\rho$ with high probability, meaning
that the amount of expended work does not exceed $w_\rho$.

The work-cost dephasing operation is necessary in general; it can in general
change the state $\rho_{E^nX^nR}$ significantly, also for permutation-invariant
states.  (We discuss this issue in more detail below, in
\cref{sec:WorkCostCovariance}.)

A core component of our proof is a semiuniversal version of the universal
conditional erasure protocol in Ref.~\cite{arXiv:1911.05563}.  Instead of
universally applying to all states $\rho_{EX}$, it applies to all states whose
work requirement of erasure does not exceed a fixed threshold.  This key step is
presented as \cref{thm:LemmaUnivErasureW0} in
\cref{appx:UnivImplPerIidInputCondErasureProofs}.

\section{Variable-work thermodynamic implementation of a quantum process}
\label{sec:MainResult}

Let us consider any general quantum channel $\mathcal{E}_{X\to X}$ (for now, not
necessarily time-covariant).  We aim to construct a universal, variable-work
thermodynamic implementation of $\mathcal{E}_{X\to X}^{\otimes n}$
(\cref{fig:VarWorkImplProcess}).
\begin{figure}
  \centering
  \includegraphics{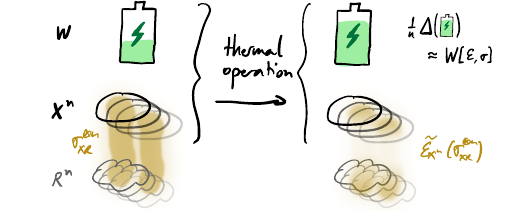}
  \caption{Variable-work implementation of any time-covariant i.i.d.\@ process
    $\mathcal{E}^{\otimes n}$ (Main Result; version for time-covariant
    processes).  Given any i.i.d.\@ input $\sigma_{XR}^{\otimes n}$, the
    implementation outputs
    $\widetilde{\mathcal{E}}_{X^n}(\sigma_{XR}^{\otimes n})$, a
    work-cost-dephased version of
    $\mathcal{E}^{\otimes n}(\sigma_{XR}^{\otimes n})$, and uses the
    asymptotically optimal work per copy $W[\mathcal{E},\sigma]$.  (For any
    time-covariant $\sigma_X^{\otimes n}$, the work-cost-dephasing is irrelevant
    without $R$, i.e.,
    $\widetilde{\mathcal{E}}_{X^n}(\sigma_{X}^{\otimes n}) =
    \mathcal{E}^{\otimes n}(\sigma_{X}^{\otimes n})$.)  Our implementation
    borrows fresh heat baths, resets them to a suitable pure state, carries out
    the unitary Stinespring dilation of $\mathcal{E}^{\otimes n}$ using these
    ancillary systems, and invokes our variable-work conditional reset operation
    to reset the ancillas to their thermal state before they are traced out.
    Accounting the work cost of each step yields the asymptotically optimal
    amount of work per copy $W[\mathcal{E},\sigma]$ when the input state is
    $\ket\sigma_{XR}^{\otimes n}$.  All ancillas and heat baths can be merged
    into a single heat bath, and all work expenditure steps can be merged into a
    single battery use; the overall process is one large thermal operation
    acting on a suitably large and accurate battery $W$, the systems $X^n$ and a
    possibly very large bath $B$.}
  \label{fig:VarWorkImplProcess}
\end{figure}

There is a tension in general between an implementation being universal for all
input states with its use of a variable amount of work.
Suppose that when given input state $\ket{\psi_1}$, the process uses an amount
$w_1$ of work, resulting in a state
$\ket{\psi_1'}\otimes\ket{-w_1}\otimes\ket{\textup{env},w_1}$, where
$\ket{\psi_1'}$ is the output state, $\ket{-w_1}$ is the state of the battery
depleted by $w_1$ energy, and $\ket{\textup{env},w_1}$ is the state of the
environment after having interacted with it spending the $w_1$ work.  Suppose
that for a different input state $\ket{\psi_2}$, an amount of work $w_2$ is
consumed resulting in the state
$\ket{\psi_2'}\otimes\ket{-w_2}\otimes\ket{\textup{env},w_2}$.  By linearity,
when given $\ket{\psi'} = \alpha_1\ket{\psi_1} + \alpha_2\ket{\psi_2}$, the
state should evolve to
$\alpha_1\ket{\psi_1'}\otimes\ket{-w_1}\otimes\ket{\textup{env},w_1} +
\alpha_2\ket{\psi_2'}\otimes\ket{-w_2}\otimes\ket{\textup{env},w_2}$.  The
implementation typically interacts differently with the environment when
different work amounts are consumed, modeled via states
$\ket{\textup{env},w_1} \perp \ket{\textup{env},w_2}$.  Then the output state on
the system and the battery is a statistical average of having consumed $w_1$
work with probability $\abs{\alpha_1}^2$ and having consumed $w_2$ work with
probability $\abs{\alpha_2}^2$.  The information about how much work was
expended is registered in the battery $W$; this causes a decoherence of the
quantum state.
More generally, a variable-work implementation typically induces a decoherence
of the process corresponding to leaking to an environment the amount of work
that was expended.  We call the resulting decohered process the
\emph{work-cost-decohered} process.

A simple example of a problematic process is
$\mathcal{E}(\cdot) = U\,`(\cdot)\, U^\dagger$ in which
$U\ket{E_1} = \ket{E_1'}$ and $U\ket{E_2} = \ket{E_2'}$, supposing that
$\ket{E_1}$, $\ket{E_2}$,$\ket{E_1'}$, $\ket{E_2'}$ are energy eigenstates with
$E_2' - E_2 \neq E_1' - E_1$.  (This process is not time-covariant; the issue
persists nevertheless for time-covariant channels, as we discuss later.)  The
implementation might wish to expend $E_1' - E_1$ work when given $\ket{E_1}$ and
$E_2' - E_2$ work when given $\ket{E_2}$, to ensure optimal usage of energy
resources, but this strategy leads to the dephasing mentioned above in case we
operate on the input $`(\ket{E_1} + \ket{E_2})/\sqrt{2}$.
For this reason, to construct a universal implementation of an arbitrary process
$\mathcal{E}_{X\to X}$ we typically expect the implementation must use a
deterministic amount of work to avoid such a dephasing.
Here, we suppose instead that we are not worried about the work-cost dephasing.
We ask for an asymptotically optimal variable-work implementation of the
process, even if it induces such a dephasing.
We further discuss the effect of this dephasing in \cref{sec:WorkCostCovariance}
below.

We now consider $\mathcal{E}_{X\to X}$ to be any arbitrary time-covariant
quantum channel.  Here, $X$ is a system with Hamiltonian $H_X$, and let
$\Gamma_X = \ee^{-\beta H_X}$.
We consider a reference system $R\simeq X$, and let
$\ket\Phi_{X:R} = \sum \ket{j}_X\otimes\ket{j}_R$ with respect to the canonical
bases of $X$ and $R$.  The \emph{Choi matrix} of $\mathcal{E}$ is
$\mathcal{E}_X(\Phi_{X:R})$.  Recall that
$(A_X\otimes\Ident_R)\ket\Phi_{X:R} = (\Ident_X\otimes A_R)\ket\Phi_{X:R}$ where
$A_R := (A_X)_R^{t}$ denotes the operator on $R$ which is the transpose of $A_X$
taken in the canonical basis.
Now consider the projector
\begin{align}
  \hat P^{(w)}_{X^nR^n} =
  \hspace*{-2.4em}
  \sum_{\substack{k,\ell,\lambda,\mu:\\
  E_\ell - \beta^{-1}S(\bar\mu) -
  E_k + \beta^{-1} S(\bar\lambda)  = w}}
  \hspace*{-3.5em}
  `\big( S_{X^n}^{E_\ell} \Pi_{X^n}^{\mu} ) \otimes
  `\big( S_{X^n}^{E_k} \Pi_{X^n}^{\lambda} )_{R^n}^t
  \,,
  \label{eq:WorkCostMeasInputOutputPhat}
\end{align}
where $`\big{ S_{X^n}^{E_k} }$ and $\Pi_{X^n}^\lambda$ are the projectors onto
the eigenspaces of $H_{X^n} = \sum_{i=1}^n H_{X_i}$ and the Schur-Weyl blocks of
$X^n$, respectively.  Here, $w$ takes values in some finite set
$\mathscr{W}\subset\mathbb{R}$ with at most $\abs{\mathscr{W}}\leq \poly(n)$
elements.  The projector $P^{(w)}_{X^nR^n}$ is designed to be applied onto a
state of the form $\mathcal{E}_X^{\otimes n}(\proj{\sigma}_{XR}^{\otimes n})$:
the part acting on $X^n$ tests the output state of the process while the
projectors on $R^n$ test the input state via its purification.  The expression
$E_\ell - \beta^{-1}S(\bar\mu)$ is an estimate for the output free energy, while
$E_k - \beta^{-1}S(\bar\lambda)$ estimates the input free energy.

The \emph{work-cost-dephased quantum channel}
$\widetilde{\mathcal{E}}_{X^n\to X^n}$ of 
$\mathcal{E}^{\otimes n}$ is defined via its Choi matrix:
\begin{align}
  \widetilde{\mathcal{E}}_{X^n\to X^n}(\Phi_{X:R})
  = \sum_{w}
  \hat P^{(w)}_{X^nR^n} \mathcal{E}_X^{\otimes n}(\Phi_{X^n:R^n}) \hat P^{(w)}_{X^nR^n}
  .
  \label{eq:WorkCostDephasedMap}
\end{align}
We call a quantum channel $\mathcal{E}_{X^n\to X^n}$ \emph{work-cost-covariant}
if $\widetilde{\mathcal{E}}_{X^n\to X^n} = \mathcal{E}_{X^n\to X^n}$.
The completely thermalizing channel (corresponding to work extraction), the
identity channel, as well as the preparation channel for any i.i.d.\@
time-covariant quantum state are all work-cost-covariant channels.
The work-cost-dephased quantum channel of any $\mathcal{E}^{\otimes n}$ is
itself work-cost-covariant.

The projectors $\hat{P}_{X^nR^n}^{(w)}$ are closely related to the
$P_{E^nX^n}^{(w)}$ defined in~\eqref{eq:WorkCostMeasEnvXpP} via a Stinespring
representation of $\mathcal{E}^{\otimes n}$ into environment systems $E^n$.  In
fact, the action of the work-cost-dephased process
$\widetilde{\mathcal{E}}_{X^n}$ can be expressed as first applying a Stinespring
dilation of $\mathcal{E}^{\otimes n}$, followed by the
dephasing~\eqref{eq:WorkCostDephasingMapEX}.  We refer to
\cref{appx:UnivImplPerIidInputProofs} for a detailed discussion and proofs.
This argument also proves that the object $\widetilde{\mathcal{E}}_{X^n}$
defined above is indeed a valid quantum channel.

We can now state our main result for a general time-covariant quantum channel:

\begin{theorem}[\textbf{Main Result; Quantum Process}]
  \label{thm:UnivImplPerIidInputCost}
  Let $\mathcal{E}_{X\to X}$ be any time-covariant quantum channel.  Let
  $\delta>0$.  Then there exists a thermal operation $\mathcal{T}_{X^nW}$ acting
  on $X^n$ and a battery $W$ along with an $\eta>0$ such that:
  \begin{enumerate}[label=(\roman*)]
  \item \label{item:thmUnivImplPerIidInputCostFidelityBound}
    For any quantum state $\sigma_{X^n}$, with a purification
    $\ket\sigma_{X^nR}$, we have
    \begin{multline}
      F`\Big(
      \tr_W`\big[ \mathcal{T}_{X^nW}`\big(\sigma_{X^nR} \otimes \tau_W^{(E_0)}) ]
      \,,\;
      \widetilde{\mathcal{E}}_{X^n\to X^n}`\big(\sigma_{X^nR})
      )
      \\[.5ex]
      \geq 1 - \poly(n)\,\ee^{-n\eta}\ ,
    \end{multline}
    where $\widetilde{\mathcal{E}}_{X^n\to X^n}$ is the work-cost-dephased
    process associated with $\mathcal{E}_{X}^{\otimes n}$.

    Furthermore, any time-covariant $\sigma_X$, then
    $\widetilde{\mathcal{E}}_{X^n}(\sigma_X^{\otimes n}) = \mathcal{E}^{\otimes
      n}(\sigma_X^{\otimes n})$;
  \item \label{item:thmUnivImplPerIidInputCostVariableWorkCostBound}
    For any quantum state $\sigma_{X}$, we have
    \begin{multline}
      \tr`\Big[
      \Pi_W^{\geq (E_0 - w_\sigma)}
      \mathcal{T}_{X^nW}`\big(\sigma_{X}^{\otimes n}\otimes\tau_W^{(E_0)})
      ]
      \\
      \geq 1 - \poly(n)\,\ee^{-n\eta}\ ;
    \end{multline}
    where
    \begin{align}
      w_\sigma
      &= n \bigl[ \beta^{-1}\DD{\mathcal{E}(\sigma_X)}{\Gamma_X}
       - \beta^{-1}\DD{\sigma_X}{\Gamma_X}
      + 7\beta^{-1}\delta \bigr]
      \nonumber\\
      &\quad+ \log\poly(n)\ .
    \end{align}
  \end{enumerate}
\end{theorem}

The proof of \cref{thm:UnivImplPerIidInputCost} is presented in
\cref{appx:UnivImplPerIidInputProofs}.  

The condition~\ref{item:thmUnivImplPerIidInputCostFidelityBound} ensures that
the implementation always accurately reproduces the target channel (which is the
work-cost-dephased channel of $\mathcal{E}^{\otimes n}$), for all i.i.d.\@ input
states.  The
condition~\ref{item:thmUnivImplPerIidInputCostVariableWorkCostBound} ensures
that whenever the input state is $\sigma_X^{\otimes n}$, the battery is never
depleted by more than a amount certain $w_\sigma$; This amount is the work
expended in the process, and it asymptotically matches the optimal per-input
state work cost per copy of
$\beta^{-1}`\big[\DD{\mathcal{E}(\sigma)}{\Gamma} - \DD{\sigma}{\Gamma}]$ after
taking $\delta\to0$.

The work cost of the thermal operation provided by
\cref{thm:UnivImplPerIidInputCost} can also be quantified for general input
states, even non-i.i.d.\@ states.  In such cases, the variable work cost is best
quantified by invoking \cref{thm:UnivCondErasureForEachIidInput} directly, as
explained below.

The central ingredient in our proof is the universal, variable-work-cost
protocol for conditional erasure from \cref{sec:ConditionalErasureVariable}
(\cref{thm:UnivCondErasureForEachIidInput}).  The strategy for the proof of
\cref{thm:UnivImplPerIidInputCost} is to construct the protocol as follows:
(1)~Borrow systems $E^n$, with carefully chosen Hamiltonians, initialized in
their thermal state (a free operation); (2)~Reset the $E^n$ systems to their
$\ket0_E$ state, costing $-nF_E$ work; (3)~Apply the Stinespring dilation
unitary $V_{XE}^{\otimes n}$ of $\mathcal{E}^{\otimes n}$, where the dilation
along with $H_E$ can be chosen such that $V_{XE}$ is energy-conserving;
(3)~Apply our conditional reset protocol to thermalize the $E^n$ systems
conditioned on $X^n$; (4)~Discard the $E^n$ systems.
The variable work cost is governed by the conditional reset step, and a careful
accounting gives the amounts stated in \cref{thm:UnivImplPerIidInputCost}.  The
dephasing operation $\mathcal{D}^{\textup{W}}_{E^nX^n}$ induced during the
conditional erasure maps exactly onto the work-cost process dephasing that
results in implementing the process $\widetilde{\mathcal{E}}_{X^n}$.

\section{The consequences of revealing the amount of work used: Work cost
  covariance}
\label{sec:WorkCostCovariance}

An implementation that uses a variable amount of work interacts differently with
the environment depending on how much work was actually used, and furthermore
leaves information about how much work was actually used in the battery.  These
correlations cause decoherence between certain states in the process (see
\cref{fig:WorkCostDephasing}).
\begin{figure}
  \centering
  \includegraphics{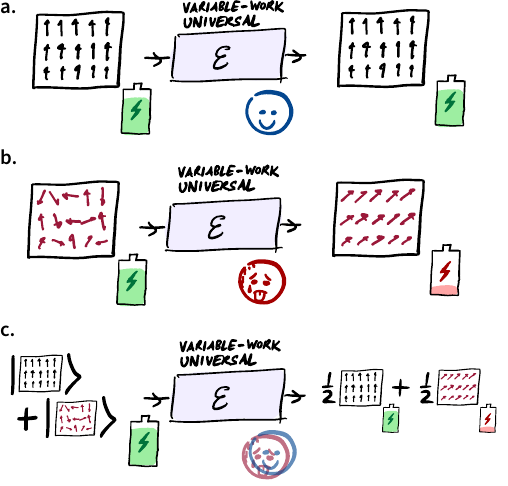}
  \caption{A variable-work cost implementation of a process necessarily induces
    decoherence, because the implementation interacts differently with the
    environment when different amounts of work are required.  In general, this
    dephasing can be significant in the presence of a reference system (not
    depicted).  For time-covariant and permutation-invariant inputs, the
    dephasing has no effect if we ignore the reference system.
    \textbf{a.}~Example variable-work implementation of a process that aligns
    spins along some input-dependent axis (for illustration purposes).  If the
    spins are already aligned, the process costs no work. \textbf{b.}~The same
    process might consume a lot of work on a different input
    state. \textbf{c.}~A superposition of different types of input states might
    decohere into a mixture because the process interacts differently with the
    heat bath for those input state, which is reflected in the different battery
    output state.  Note the work-cost-dephasing is significantly subtler than a
    dephasing of the input or output state in the energy basis: For instance,
    the identity process suffers no work-cost-dephasing, and is implemented
    accurately by our construction for arbitrary input states, because all input
    states require the same amount of work (zero work).}
  \label{fig:WorkCostDephasing}
\end{figure}
For this reason, if we require an accurate implementation of the process for all
input states including non-i.i.d.\@ states, it is in general necessary to use a
deterministic, rather than variable, amount of work.
We now briefly study the effect of this work-cost dephasing.

One could imagine the work cost dephasing corresponds to dephasing the input or
output states in the energy basis.  This is not the case.  The identity process,
for instance, is work-cost covariant because all input states necessitate the
same amount of work (zero work).  In this case, there is only a single term
\eqref{eq:WorkCostDephasedMap} and
$\widetilde{\mathcal{E}}_{X^n} = \mathcal{E}_X^{\otimes n}$.

A main question we can ask is, do there even exist time-covariant processes that
are not work-cost-covariant?
%
%
Here, we focus on the task of conditional erasure of a time-covariant state
$\rho_{E^nX^n}$ where $H_{EX} = H_E + H_X$.  As it turns out, there are states
$\ket\rho_{E^nX^nR}$ for which the work-cost-dephased version
$\tr_{E^n}`\big[\mathcal{D}^{(\textup{W})}_{E^nX^n}(\rho_{E^nX^nR})]$, after
conditional erasure of $E^n$, is significantly different from the reduced state
$\rho_{X^nR}$. [Recall the definition of $\mathcal{D}^{(\textup{W})}_{E^nX^n}$
in \cref{eq:WorkCostDephasingMapEX}].  This shows that the process
$\mathcal{E}_{\textup{cond.reset}\,EX}$ that formalizes the conditional erasure
task is not work-cost-covariant,
$\widetilde{\mathcal{E}}_{X^n} \neq \mathcal{E}^{\otimes n}$.  This argument,
including a detailed construction of the example, is presented in
\cref{appx:WorkCostDephasingDetails}.  The example is constructed in a setting
with trivial Hamiltonians, i.e. $H_X=0$, $H_E=0$.

Interestingly, the work-cost dephasing only affects the reduced state
$\rho_{X^nR}$ in its correlations between $X^n$ and $R$, if $\rho_{E^nX^nR}$ is
time-covariant and permutation-invariant:
\begin{proposition}
  \label{thm:WrkCostDephasBipartPINoRef}
  Let $\rho_{E^nX^n}$ be a quantum state that is time covariant and permutation
  invariant over the copies $(EX)^n$.  Then
  \begin{align}
    \tr_{E^n}`\Big{ \mathcal{D}^{(\textup{W})}_{E^nX^n}(\rho_{E^nX^n}) }
    = \rho_{X^n}\ .
  \end{align}
\end{proposition}

The corresponding statement for quantum processes implies that the work-cost
dephasing has no effect if the goal is to implement a state transformation (with
no regards to preserving correlations with a reference system).

\begin{proposition}
  \label{thm:WrkCostDephasProcessPINoRefStateTransf}
  Let $\widetilde{\mathcal{E}}_{X^n\to X^n}$ be the work-cost-dephased map
  associated with a i.i.d.\@ quantum channel $\mathcal{E}_{X}^{\otimes n}$ where
  $\mathcal{E}$ is time covariant.  Let $\sigma_{X^n}$ be any time-covariant,
  permutation-invariant state.  Then
  \begin{align}
    \widetilde{\mathcal{E}}_{X^n\to X^n}(\sigma_{X^n})
    = \mathcal{E}_{X\to X}^{\otimes n}(\sigma_{X^n})\ .
  \end{align}
\end{proposition}

The proofs of \cref{thm:WrkCostDephasBipartPINoRef,%
  thm:WrkCostDephasProcessPINoRefStateTransf} are presented in
\cref{appx:WorkCostDephasingDetails}.

\section{Outlook}
\label{sec:Outlook}

We constructed a \emph{variable-work}, universal implementation of any
time-covariant $\mathcal{E}_{X\to X}^{\otimes n}$, up to a dephasing operation
amounting to telling an environment how much work was expended.
The amount of work consumed by the implementation, when provided the input state
$\sigma_X^{\otimes n}$, matches asymptotically the minimal amount of work per
copy that even an implementation tailor-made for $\sigma_X$ would have to
expend.
Furthermore, for any time-covariant $\sigma_X$, the implementation's output is
close to $\mathcal{E}^{\otimes n}(\sigma_X^{\otimes n})$.

This implementation fits between two extremes.  An implementation of a
time-covariant $\mathcal{E}^{\otimes n}$ tailored to a specific input state
$\rho$ can implement the process accurately using an asymptotically optimal
amount of work $W[\mathcal{E},\rho]$.  Its output is generally unreliable if any
input state other than $\rho$ is provided.  On the other hand, a universal
implementation of a time-covariant $\mathcal{E}^{\otimes n}$ can reproduce the
process' output accurately for all input states $\sigma_{X^nR}$, including
non-i.i.d.\@ states, if it uses a deterministic, worst-case asymptotic work cost
per copy $\max_\sigma W[\mathcal{E},\sigma]$ regardless of the provided input
state.
Our variable-work-cost implementation lies between these two situations: The
implementation uses an asymptotically optimal amount of work per copy
$W[\mathcal{E},\sigma]$ for any i.i.d.\@ input $\sigma^{\otimes n}$; our
implementation also has a well-characterized output for all input states,
including non-i.i.d.\@ ones, although the output deviates from that of the
desired $\mathcal{E}^{\otimes n}$ by some dephasing corresponding to revealing
the amount of work used.

The variable-work implementation is derived by piecing together tools and
techniques introduced in our earlier
Refs.~\cite{Faist2019PRL_thcapacity,Faist2021CMP_impl}.  We have also taken the
opportunity to clarify how some results presented there immediately yield
implementations that are universal to specific subsets of input states
(\cref{appx:UniversalWithSetOfInputs}).
We summarize the thermodynamic implementation statements made throughout this
paper in \cref{appx:SummaryStatements}, to help the reader navigate this
document.

Our main result applies to any arbitrary time-covariant quantum process
$\mathcal{E}_{X\to X}$, which can be any system evolution that is compatible by
quantum mechanics.  This includes work extraction, the preparation of a quantum
state, running a quantum computation, or even acquisition of information by a
probe system (for instance, in the context of quantum
metrology~\cite{LipkaBartosik2018JPA_thermodynamic}).
Furthermore, we show in \cref{appx:UnivVariableImplWithGPM} (cf.\@
\cref{thm:UnivVarWorkIidGPM} there) that our main result extends to
non-time-covariant processes in case we allow arbitrary Gibbs-preserving maps
instead of thermal operations.

As a key ingredient in our proof, we construct a variable-work-cost protocol for
\emph{conditional erasure} of a system using a memory.  We construct this
protocol by following closely ideas and concepts introduced in our earlier
Refs.~\cite{Faist2019PRL_thcapacity,Faist2021CMP_impl}.
As discussed in \cref{sec:ConditionalErasureVariable}, the task of conditional
erasure to some general state $\kappa_E$ naturally encompasses several
thermodynamic tasks whose connection might not immediately be obvious to all,
including \emph{complete thermalization}, \emph{conditional thermalization}, as
well as \emph{work extraction} or even \emph{conditional work extraction}.

We can also consider the task of \emph{state preparation}.  Let $\rho_{X}$ be
any quantum state, which we assume time-covariant so it can be prepared using
thermal operations and work.  The channel
$\mathcal{E}_{\mathrm{prep.}\,\rho}(\cdot) = \tr(\cdot)\rho_X$ prepares the
state $\rho_X$; this channel is work-cost-covariant because $\rho^{\otimes n}$
commutes on one hand with the total energy (it is time covariant) as well as
with the Schur-Weyl blocks (it is permutation-invariant).
\Cref{thm:UnivImplPerIidInputCost} then provides a variable-work-cost of this
state preparation of this channel, with an asymptotic work cost per copy of
$w_\sigma = \beta^{-1}`\big[ \DD{\rho_X}{\Gamma_X} - \DD{\sigma_X}{\Gamma_X} ] =
\beta^{-1}`\big[ F_X(\rho_X) - F_X(\sigma_X) ]$, where $\sigma_X$ is the initial
state on $X$.  This cost is optimal, as it coincides with the optimal asymptotic
work cost of the state transformation
$\sigma_X^{\otimes n} \to \rho_X^{\otimes n}$~\cite{%
  Brandao2013_resource,Horodecki2013_ThermoMaj}.

While our results assume $\mathcal{E}$ to be time-covariant, no such assumption
applies to the implementation's input states.  That is, our implementation
performs accurately also for input quantum states that fail to commute with the
system Hamiltonian.  Furthermore, the work-cost-dephasing need not decohere
input states into the energy eigenbasis.  As a trivial example, the identity
channel can be implemented by doing nothing and expending no work, which is
optimal for each input state; states with coherent superpositions of energy
levels are not decohered.  The work-cost-dephasing does not affect the identity
channel since there is only a single work cost event where no work is expended,
so there is only a single term in the sum in~\eqref{eq:WorkCostDephasedMap}.
If the input state is time-covariant and permutation-invariant, the work-cost
dephasing has no effect if we disregard the reference system $R$.  However, even
for systems with trivial Hamiltonians, we can construct explicit examples of
input states $\sigma_{X^nR}$ that are strongly affected by this dephasing.
A deeper study of the connection between time-covariance and
work-cost-covariance is likely to yield close connections to known
symmetry-based decompositions of quantum processes into symmetry
modes~\cite{Cirstoiu2017arXiv_gauge,Cirstoiu2020PRX_robustness} and how
representations of the symmetric group action that act on bipartite systems
split into irreducible representations on each
party~\cite{Christandl2006CMP_spectra,PhDChristandl2006_bipartite}.

Should one need to avoid excessive work cost dephasing, one could settle for an
intermediate regime where the possible amount of work expended is suitably
coarse-grained so as not to reveal too drastically the input state, thereby
reducing decoherence.  This can likely be achieved by choosing a coarser
measurement of the anticipated work cost $w$, and would correspond to
``smudging'' or ``smoothing'' the work-cost measurement to avoid strong
work-cost dephasing.

The framework of thermal operations can be formulated in a consistent way also
for processes whose input and output system dimensions differ; for instance,
energy-conserving partial isometries can be extended to energy-conserving
unitaries on those systems augmented by additional auxiliary systems in standard
states, leading to a natural definition of a ``generalized thermal operation''
in this regime (cf.\@ e.g.\@ Ref.~\cite{Sagawa2021JPA_asymptotic}).
We anticipate that our results naturally extend to this regime.

We focus on finite-dimensional settings for technical simplicity.  Remarkably,
Watanabe \emph{et al.}\@ show that their techniques for implementing a work
extraction protocol with variable work yield extend to infinite-dimensional
systems~\cite{Watanabe2025arXiv_universal}.  It is unclear if the techniques
presented here can also be extended to this regime.

Finally, these results are likely to extend to regimes beyond the fully i.i.d.\@
setting, including for ergodic spin chains and translation-invariant
states~\cite{Faist2019PRL_macroscopic,Sagawa2021JPA_asymptotic}.

\begin{acknowledgments}
The author thanks 
Kaito Watanabe,
Ryuji Takagi,
Bartosz Regula,
Johannes J.\@ Meyer,
and 
David Jennings
for discussions and feedback.
\end{acknowledgments}

\vfill\hbox{}
\clearpage
\onecolumngrid
\vspace{3cm}
\begin{center}%
\par\noindent\hbox{\hfill\fbox{\vbox{\textbf{APPENDIX}}}\hfill}\par\end{center}
\vspace{.5cm}
\appendix
\twocolumngrid

\section{Summary of the constructions presented in this paper}
\label{appx:SummaryStatements}

Because this paper grew rapidly from a short focused note to a collection of
multiple technical statements spanning several settings, we're summarizing the
statements made in this paper here:

\begin{itemize}
\item Warm-Up (Work Extraction): Universal work extraction with variable work
  cost~\cite{Watanabe2025arXiv_universal}, with an alternative construction
  pieced together from the techniques of Ref.~\cite{Faist2021CMP_impl}.
  $\to$~\cref{sec:WarmupWorkExtraction}

\item Main Result (Conditional Erasure): Universal implementation of conditional
  erasure to the thermal state, conditioned on a quantum memory, up to a
  work-cost dephasing, using thermal operations, using a \emph{variable} amount
  of work which is asymptotically optimal for each i.i.d.\@ input state.
  $\to$~\cref{thm:UnivCondErasureForEachIidInput} in
  \cref{sec:ConditionalErasureVariable}

\item Main Result (Quantum Process): Universal implementation of any arbitrary
  time-covariant i.i.d.\@ process $\mathcal{E}^{\otimes n}$, up to a work-cost
  dephasing, using thermal operations, using a \emph{variable} amount of work
  which is asymptotically optimal for each i.i.d.\@ input state.
  $\to$~\cref{thm:UnivImplPerIidInputCost} in \cref{sec:MainResult}

\item (Clarification regarding Ref.~\cite{Faist2021CMP_impl}:) Implementation of
  any arbitrary time-covariant i.i.d.\@ process $\mathcal{E}^{\otimes n}$, using
  thermal operations, universal with respect to an arbitrary set of states
  $\hat{\mathscr{S}}_X$, with a deterministic amount of work given by the worst
  case over that set.
  $\to$~\cref{thm:UniversalTheoremCovariantInputSet} in \cref{appx:UniversalWithSetOfInputs}

\item (Clarification regarding Ref.~\cite{Faist2021CMP_impl}:) A similar
  statement, but for arbitrary processes $\mathcal{E}_{X\to X'}$ and using
  Gibbs-preserving maps.
  $\to$~\cref{thm:UniversalTheoremGpmInputSet} in \cref{appx:UniversalWithSetOfInputs}

\item Main Proof Technique: \emph{Semiuniversal} implementation of conditional
  erasure to the thermal state, conditioned on a quantum memory, up to a
  work-cost dephasing, using thermal operations, using a deterministic amount of
  work $w$, universal with respect to all i.i.d.\@ states whose conditional
  erasure work cost does not exceed $w$.
  $\to$~\cref{thm:LemmaUnivErasureW0} in
  \cref{appx:UnivImplPerIidInputCondErasureProofs}

\item Gibbs-preserving maps: Universal implementation of any arbitrary i.i.d.\@
  process $\mathcal{E}^{\otimes n}$ (need not be time-covariant), up to a
  work-cost dephasing, using Gibbs-preserving maps, using a \emph{variable}
  amount of work which is asymptotically optimal for each i.i.d.\@ input state.
  $\to$~\cref{thm:UnivVarWorkIidGPM} in \cref{appx:UnivVariableImplWithGPM}

\item Work-cost-dephasing typically irrelevant without the reference system $R$:
  Both in the case of conditional erasure as well as in the implementation of a
  time-covariant process, if we ignore the reference system $R$ and if the input
  state is time covariant and permutation invariant, then the
  work-cost-dephasing has no effect.
  $\to$~\cref{thm:WrkCostDephasBipartPINoRef,%
    thm:WrkCostDephasProcessPINoRefStateTransf} in
  \cref{sec:WorkCostCovariance}

\item Input states that are strongly affected by the work-cost-dephasing: We
  find an explicit example of an input state whose correlations with $R$ are
  drastically altered by the work-cost-dephasing.  This dephasing, therefore, is
  necessary to accurately describe the action of our variable-work-cost
  implementation.
  $\to$~\cref{appx:UnivVariableImplWithGPM}
\end{itemize}

All statements except the last one are proven using thermal operations and
assume the process to be time-covariant.

\section{Universal implementation of any time-covariant i.i.d.\@ process
  with deterministic work for arbitrary sets of input states}
\label{appx:UniversalWithSetOfInputs}

Here, we adapt Theorem~7.1 of Ref.~\cite{Faist2021CMP_impl}, relaxing the
diamond norm requirement to a particular subset of all input states.
At a high level, the only adaptation required in the proof is to use the
universal typical conditional projector not associated with the worst-case work
cost over all inputs (which gives the thermodynamic capacity), but over the
particular set of states that we wish to consider.  The rest of the proof is
straightforward.

\begin{theorem}
  \label{thm:UniversalTheoremCovariantInputSet}
  Let $X$ be a quantum system, $H_X$ be a Hamiltonian, $\beta>0$, and
  $\mathcal{E}_{X\to X}$ a time-covariant completely positive, trace-preserving
  map.  Let $\hat{\mathscr{S}}_X$ be an arbitrary set of quantum states on $X$.
  Write $\Gamma_X = \ee^{-\beta H_X}$.
  Let $\delta>0$ small enough and $n\in\mathbb{N}$ large enough.  Then there
  exists a thermal operation $\Phi_{X^nW}$ along with battery states
  $\tau_W^{(E)}$, $\tau_W^{(E')}$, and a $\eta>0$, such that:
  \begin{enumerate}[(i)]
  \item The total work cost satisfies
    \begin{align}
      E - E'
      &= n \max_{\sigma_X\in\hat{\mathscr{S}}_X} 
      \beta^{-1} \bigl[ \DD{\mathcal{E}_X(\sigma_X)}{\Gamma_X}
        -
        \DD{\sigma_X}{\Gamma_X}
      \bigr]
      \nonumber\\[.5ex]&\quad\ 
        + 5n\beta^{-1}\delta + O(1)\ ;
      \label{eq:UniversalTheoremCovariantInputSetWorkCost}
    \end{align}
  \item Let $\sigma_{X^nR}$ be any quantum system on $X^n$ and a reference system
    $R$ such that $\sigma_{X^n} = \tr_R(\sigma_{X^nR})$ is in the convex hull of
    $`{ \sigma_X^{\otimes n}\ :\ \sigma_X\in\hat{\mathscr{S}}_X}$.  Then
    \begin{multline}
      \frac12\onenorm[\big]{
      \tr_W`[\Phi_{X^nW}(\sigma_{X^nR}\otimes\tau_W^{(E)})]
      - 
      \mathcal{E}_X^{\otimes n}(\sigma_{X^nR})
      }
      \\
      \leq \poly(n)\,\ee^{-n\eta}\ .
    \end{multline}
  \end{enumerate}
\end{theorem}

\begin{proof}[**thm:UniversalTheoremCovariantInputSet]
  The proof proceeds exactly as the proof of Theorem~7.1 in
  Ref.~\cite{Faist2021CMP_impl}, with the following adaptations.  Let
  $V_{XE}$, $\ket{0}_E$, $H_E$, $F_E$, $Z_E$, be as in the original proof, and
  let $\Gamma_E=\ee^{-\beta H_E}$,
  $\Gamma_{XE} = \Gamma_X\otimes\Gamma_E = \ee^{-\beta(H_X+H_E)}$.  Let
  \begin{align}
    x = \min_{\sigma_X\in\hat{\mathscr{S}}_X} `\big[
     \DD{\sigma_X}{\Gamma_X}
    - \DD{\mathcal{E}(\sigma_X)}{\Gamma_X}
    ] \geq \beta F_E\ .
  \end{align}
  For any $\sigma_X$, let
  $\rho_{XE}(\sigma_X) = V_{XE}(\proj0_E\otimes\sigma_X)V_{XE}^\dagger$ be as in
  the original proof, and let
  \begin{align}
    \mathscr{S}_{E^nX^n} = `\Big{
    [\rho_{EX}(\sigma_X)]^{\otimes n}\ :\ \sigma_X\in\mathscr{S}_X
    }\ .
  \end{align}
  For all $\rho_{EX}^{\otimes n}\in\mathscr{S}_{E^nX^n}$, we have by
  definition of $x$ that
  \begin{align}
    \DD{\rho_{EX}(\sigma_X)}{\Gamma_{XE}}
    - 
    \DD{\rho_{X}(\sigma_X)}{\Gamma_{X}}
    \geq x\ .
  \end{align}
  As in the original proof, let $P_{E^nX^n}^{x,\delta}$ be the universal typical
  and relative conditional operator provided by Proposition~6.1 of
  Ref.~\cite{Faist2021CMP_impl}, noting as in the original proof that
  $P_{E^nX^n}^{x,\delta}$ is a projector commuting with $\Gamma_X^{\otimes n}$
  and $\Gamma_{EX}^{\otimes n}$.  As in the original proof, we have
  \begin{align}
  \tr`\big[ P_{E^nX^n}^{x,\delta} \rho_{EX}^{\otimes n} ]
    &\geq 1 - \poly(n)\,\ee^{-n\eta}
      \nonumber\\
    \tr`\big[P_{E^nX^n}^{x,\delta} \rho_X^{\otimes n}\otimes\gamma_E^{\otimes n}]
    &\leq \poly(n)\,\ee^{-n(x - \beta F_E - 4\delta)}\ ,
  \end{align}
  for all $\rho_{EX}\in\mathscr{S}_{E^nX^n}$, for some $\eta>0$ independent of
  $\rho$ and $n$, and where $\gamma_E=\Gamma_E/\tr(\Gamma_E)$.  We assume
  $\eta\leq\delta$, replacing $\eta$ by $\delta$ if necessary without
  consequence for the proof.  As in the original proof, let $J$ be a quantum
  register, $m = \log \lfloor \ee^{n(x - \beta F_E - 5\delta)} \rfloor$, and let
  $\mathscr{R}_{E^nX^nJ}$ be the thermal operation furnished by Proposition~7.1
  of Ref.~\cite{Faist2021CMP_impl} for $\kappa,\kappa'=\poly(n)\,\ee^{-n\eta}$.
  (We assume $x-\beta F_E - 5\delta > 0$, or else a trivial thermal operation
  that extracts no work can be used as in the original proof.)  The thermal
  operation extracts an amount of work, via the $J$ register, given by
  \begin{align}
    \beta^{-1} m
    &= \beta^{-1} n(x - \beta F_E - 5\delta) + \nu\ ,
  \end{align}
  where $0\leq \nu\leq\log(2)$ accounts for a possible rounding error due to the
  floor function.  As in the original proof, accounting for the work cost $-nF_E$
  to prepare the ancillary systems $E$ in their $\ket0_E$ state, the total work
  cost of the protocol is then exactly the claimed
  amount~\eqref{eq:UniversalTheoremCovariantInputSetWorkCost}.
  It is left to prove that the protocol performs accurately for all desired
  input states.  But this property follows immediately from the guarantees on
  the conditional erasure thermal operation $\mathcal{R}_{E^nX^nJ}$ provided by
  Proposition~7.1 of Ref.~\cite{Faist2021CMP_impl}.
\end{proof}

A similar adaptation is possible for Theorem~6.1 of
Ref.~\cite{Faist2021CMP_impl}, relaxing the diamond norm requirement to a
particular subset of all input states.  This provides an implementation of any
arbitrary i.i.d.\@ process (not necessarily time-covariant), based on
Gibbs-preserving maps, and which is universal for a particular set of input
states.

\begin{theorem}
  \label{thm:UniversalTheoremGpmInputSet}
  Let $\Gamma_X,\Gamma_{X'}>0$, $\mathcal{E}_{X\to X'}$ be an arbitrary cp.\@
  tp.\@ map, and let $\epsilon>0$.  Let $\hat{\mathcal{S}}_X$ be an arbitrary
  set of states on $X$.  Then for $\delta>0$ and $n\in\mathbb{N}$ sufficiently
  large, there exists a completely positive, trace-nonincreasing map
  $\mathcal{T}_{X^n\to X'^n}$ satisfying:
  \begin{enumerate}[label=(\roman*)]
  \item
    $\mathcal{T}_{X^n\to X'^n}`(\Gamma_X^{\otimes n}) \leq
    \ee^{n[y+4\delta+n^{-1}\log\poly(n)]}\Gamma_{X'}^{\otimes n}$ with
    \begin{align}
      y = \max_{\sigma\in\hat{\mathscr{S}}_X}`\big[
      \DD{\mathcal{E}_{X\to X'}(\sigma_X)}{\Gamma_{X'}} - \DD{\sigma_X}{\Gamma_X} ]\ .
    \end{align}
  \item For any $\sigma_{X^nR}$ such that $\sigma_{X^n}$ lies in the convex hull
    of $`\big{ \sigma^{\otimes n}\;:\; \sigma\in\hat{\mathscr{S}}_X }$, we have
    \begin{align}
      \frac12\onenorm[\big]{
      \mathcal{T}_{X^n\to X'^n}(\sigma_{X^nR})
      - \mathcal{E}_{X\to X'}^{\otimes n}(\sigma_{X^nR})
      } \leq \epsilon\ .
    \end{align}
  \end{enumerate}
\end{theorem}
\begin{proof}[**thm:UniversalTheoremGpmInputSet]
  The proof proceeds exactly as the proof of Theorem~6.1 in
  Ref.~\cite{Faist2021CMP_impl}, with the following adaptations.  Let
  \begin{align}
    x = -\max_{\sigma_X\in\hat{\mathscr{S}}_X} `\Big{
    \DD{\mathcal{E}(\sigma_X)}{\Gamma_X} - \DD{\sigma_X}{\Gamma_X}
    }\ .
  \end{align}
  The operator $M_{E^nX^n}^{x,\delta}$ constructed in the original proof thanks
  to Proposition~6.1 of Ref.~\cite{Faist2021CMP_impl} has exactly the desired
  properties for all $\ket\rho_{X'ER}$ obtained from any $\ket\sigma_{XR}$ with
  $\sigma_X\in\hat{\mathscr{S}}_X$, by construction.
  Let $V_{X\to X'E}$ and
  $W_{X^n\to X'^nE^n} = M_{E^nX'^n}^{x,\delta} V_{X\to X'E}^{\otimes n}$ be as
  in the original proof.
  Let $\sigma_{X^n} = \sum p_j \sigma_X^{(j)}$ with
  $\sigma_X^{(j)}\in\hat{\mathscr{S}}_X$ be any convex combination of i.i.d.\@
  states from $\hat{\mathscr{S}}_X$.  Let $R'$ be a quantum register such that
  we can construct a purification of $\sigma_{X^n}$ as
  \begin{align}
    \ket{\sigma}_{X^nR^nR'}
    = \sum_j \, \sqrt{p_j} \, \ket{\sigma^{(j)}}_{XR}^{\otimes n}
    \otimes\ket{j}_{R'}\ ,
  \end{align}
  where $\ket{\sigma^{(j)}}_{XR}$ is a purification of $\sigma_X^{(j)}$.
  Using Uhlmann's theorem,
  \begin{align}
    \hspace*{1em}&\hspace*{-1em}
    F`\Big( \mathcal{E}_{X\to X'}^{\otimes n}(\sigma_{X^nR^nR'}) ,
                   \mathcal{T}_{X^n\to X'^n}(\sigma_{X^nR^nR'}) )
    \nonumber\\
    &\geq\Re`\Big{
    \bra{\sigma}_{X^nR^nR'} (V^\dagger)^{\otimes n}
    W_{X^n\to X'^nE^n} \ket{\sigma}_{X^nR^nR'}
    }
    \nonumber\\
    &=
    \Re`\Big{
    \bra{\sigma}_{X^nR^nR'} (V^\dagger)^{\otimes n}
      M_{E^nX'^n}^{x,\delta}
      (V)^{\otimes n} \ket{\sigma}_{X^nR^nR'}
    }
    \nonumber\\
    &= \sum_j p_j
    \Re`\Big{
    \bra{\sigma^{(j)}}_{XR}^{\otimes n} (V^\dagger)^{\otimes n}
      M_{E^nX'^n}^{x,\delta}
      (V)^{\otimes n}
      \ket{\sigma^{(j)}}_{XR}^{\otimes n}
    }
    \nonumber\\
    &\geq 1 - \poly(n)\,\ee^{-n\xi}\ ,
  \end{align}
  where $\xi>0$ independent of $\sigma$ and $n$ is given to us from
  Proposition~6.1 of Ref.~\cite{Faist2021CMP_impl}.
\end{proof}

\section{Universal conditional work-cost-dephased erasure
  with optimal per-i.i.d.\@ input work cost}
\label{appx:UnivImplPerIidInputCondErasureProofs}

In this appendix, we focus on the task of conditional
erasure~\cite{delRio2011Nature,Faist2015NatComm}, and prove
\cref{thm:UnivCondErasureForEachIidInput} of
\cref{sec:ConditionalErasureVariable}.  As a reminder, this conditional erasure
protocol is the key step to proving our main result for general time-covariant
processes in \cref{sec:MainResult}.

A key step to proving \cref{thm:UnivCondErasureForEachIidInput} is a
semiuniversal version of the conditional erasure protocol of our earlier
Ref.~\cite{arXiv:1911.05563}.  The statement is formalized in the following
lemma.  It is essentially an immediate consequence of Proposition~7.1 and
Proposition~6.1 with a fixed work threshold.

To state our lemma, we define for any $w_0$ the projector onto all the outcomes
of $P^{(w)}_{E^nX^n}$ associated with $w\leq w_0$:
\begin{align}
  P^{(\leq w_0)}_{E^nX^n} = \sum_{w\leq w_0} P_{E^nX^n}^{(w)}\ .
\end{align}

\begin{lemma}
  \label{thm:LemmaUnivErasureW0}
  Let $E,X$ be quantum systems with Hamiltonians $H_E$, $H_X$, and let
  $\Gamma_X = e^{-\beta H_X}$, $\Gamma_E = e^{-\beta H_E}$,
  $\Gamma_{XE} = \Gamma_X\otimes\Gamma_E$ and $\gamma_E=e^{-\beta H_E}$.  Let
  $\delta>0$ and $w_0 \in \mathbb{R}$.  Then there exists a thermal operation
  $\mathcal{R}^{(w_0)}_{E^nX^nW}$ acting on a battery $W$ and a $\eta>0$ such
  that for any state $\hat\rho_{E^nX^n}$ in the convex hull of
  \begin{align}
    \hspace*{1em}&\hspace*{-1em}
    \mathscr{S}_{E^nX^n}^{(w_0)} =
    \nonumber\\[1ex]
    &`\Big{ \rho^{\otimes n}_{EX}\,:\; \beta^{-1} `\big[
    D(\rho_X\Vert\Gamma_X) -  D(\rho_{EX}\Vert\Gamma_{EX}) ] \leq w_0 }
      \nonumber\\
    &\cup\ 
    `\Big{
    \hat\rho_{E^nX^n}\,:\:
    \tr`\big[P^{(\leq w_0)}_{E^nX^n} \hat\rho_{E^nX^n}] = 1
    }
    \ ,
    \label{eq:LemmaUnivErasureW0SetOfInputStates}
  \end{align}
  we have
  \begin{multline}
    F`\Big( 
      \mathcal{R}^{(w_0)}_{E^nX^nW}`\big[ \hat\rho_{E^nX^nR} \otimes \tau^{(E_0)}_W ]
      ,
      \gamma_E^{\otimes n} \otimes \hat\rho_{X^nR} \otimes \tau^{(E')}_W
    )
    \\[1ex]
    \geq 1 - \poly(n)\,e^{-n\eta}\ ,
    \label{eq:LemmaUnivErasureW0Correct}
  \end{multline}
  where $\ket{\hat\rho}_{E^nX^nR}$ is a purification of $\hat\rho_{E^nX^n}$, and where
  $E_0, E'$ are such that
  \begin{align}
    E_0 - E' \leq n(w_0 + 5\beta^{-1}\delta + F_E) + \log(2)\ .
    \label{eq:LemmaUnivErasureW0WorkCost}
  \end{align}
\end{lemma}

This lemma is essentially the same as
\cref{thm:UniversalTheoremCovariantInputSet}, stated for the task of conditional
erasure, with some subtle but minor differences.  Here, the thermal operation
$\mathcal{R}^{(w_0)}_{E^nX^nW}$ resets $E$ to the thermal state, while
preserving the state $\rho_{XR}$ including correlations with purifications; it
does so by expending approximately $w_0$ work per copy, asymptotically.  It
succeeds for all states whose state-dependent erasure work cost does not exceed
$w_0$.  Not only does it successfully conditionally erase $E^n$ for any i.i.d.\@
state $\rho_{XE}^{\otimes n}$, but it also does so for convex combinations of
such states, while still preserving correlations between $X^n$ with any
purification of such states.

With this lemma, we can prove \cref{thm:UnivCondErasureForEachIidInput}:
\begin{proof}[*thm:UnivCondErasureForEachIidInput]
  The protocol will proceed as follows.  First, we perform a measurement on
  $\rho_{XE}^{\otimes n}$ that tests how much work we expect to have to spend to
  carry out a state-dependent protocol, described by the POVM
  $`\big{ P^{(w)}_{E^nX^n} }$.  Depending on the measurement result, we apply a
  corresponding semiuniversal protocol given by the lemma.

  We can now build our full protocol.  For any $w$, let
  $\mathcal{R}^{(w+\beta^{-1}\delta)}_{E^nX^nW}$ be the thermal operation
  furnished by \cref{thm:LemmaUnivErasureW0} for $w_0 = w + \beta^{-1}\delta$.
  Let
  \begin{align}
    \mathcal{R}'_{E^nX^nW} = \sum_{w\in\mathscr{W}}
    \mathcal{R}^{(w+\beta^{-1}\delta)}_{E^nX^nW}
        `\big( P^{(w)}_{X^nE^n} `(\cdot) P^{(w)}_{X^nE^n} )\ .
  \end{align}
  The process $\mathcal{R}'_{E^nX^nW}$ can be implemented by a thermal operation
  which consumes an additional amount of work $\log\poly(n)$, simply because the
  protocol applies the thermal operation
  $\mathcal{R}^{(w+\beta^{-1}\delta)}_{E^nX^nW}$ conditional on a measurement
  $P^{(w)}_{X^nE^n}$ which commutes with energy, so the overall protocol
  corresponds to a global energy-conserving controlled-unitary with a large
  bath.  The amount of work $\log\poly(n)$ arises from resetting any ancillary
  memory register that stored the measurement outcome $w$, and which serves as
  the control system of the controlled-thermal operation.  Let
  $\mathcal{R}_{E^nX^nW}$ be the global thermal operation that implements
  $\mathcal{R}'_{E^nX^nW}$, which includes an additional decrease of at most
  $\log\poly(n)$ of the battery charge in $W$.

  For any state $\ket{\rho}_{E^nX^nR}$, let
  \begin{align}
    \rho^{(w)}_{E^nX^nR}
    &= \frac1{\alpha_w} \, P^{(w)}_{E^nX^n} \rho_{E^nX^nR} P^{(w)}_{E^nX^n}\ ;
      \nonumber\\
    \alpha_w &= \tr`\Big(P^{(w)}_{E^nX^n} \, \rho_{E^nX^n})  \ .
  \end{align}
  Then
  $\tilde\rho_{E^nX^nR} =
  \mathcal{D}^{\textup{W}}_{E^nX^n}[\rho_{E^nX^nR}] = \sum_w \alpha_w
  \rho^{(w)}_{E^nX^nR}$.  We then have
  $\rho_{E^nX^n}^{(w)} \in \mathscr{S}_{E^nX^n}^{(w+\beta^{-1}\delta)}$ by
  construction (cf.\@ \cref{thm:LemmaUnivErasureW0}).  So
  \cref{thm:LemmaUnivErasureW0} tells us that
  \begin{multline}
    F`\Big(
    \tr_W`\big[\mathcal{R}^{(w+\beta^{-1}\delta)}_{E^nX^nW}`\big(
        \rho^{(w)}_{E^nX^nR} \otimes \tau_W^{(E_0)} ) ]
    \,,\;
    \rho^{(w)}_{X^nR}
    )
    \\
    \geq 1 - \poly(n)\,\ee^{-n\eta}\ .
  \end{multline}
  Therefore, by joint concavity of the fidelity,
  \begin{align}
    \hspace*{1em}&\hspace*{-1em}
    F`\Big(
    \tr_W`\big[\mathcal{R}'_{E^nX^nW}`\big(
        \rho_{E^nX^nR} \otimes \tau_W^{(E_0)} ) ]
    \,,\;
    \tilde\rho_{X^nR}
    )
    \nonumber\\
    &\geq
      \sum_{w\in\mathscr{W}} \alpha_w
      F`\bigg(
    \tr_W`\big[\mathcal{R}^{(w+\beta^{-1}\delta)}_{E^nX^nW}`\big(
        \rho^{(w)}_{E^nX^nR} \otimes \tau_W^{(E_0)} ) ]
    \,,\;
    \rho^{(w)}_{X^nR}
    )
    \nonumber\\
    &\geq 1 - \poly(n)\,\ee^{-n\eta}\ .
  \end{align}
  The ``furthermore'' part of the claim follows from
  \cref{thm:WrkCostDephasBipartPINoRef}.  This
  proves~\ref{item:thmUnivCondErasureForEachIidInputOutputFidelity}.

  We'll prove~\ref{item:thmUnivCondErasureForEachIidInputWorkCostIidInput} in a
  moment; let's first tackle
  point~\ref{item:thmUnivCondErasureForEachIidInputWorkCostAnyInput}.  Let
  $\rho_{E^nX^n}$ be any quantum state.  We write for short  as above
  $\rho^{(w)}_{E^nX^nR} = P^{(w)}_{E^nX^n} \rho_{E^nX^nR}
  P^{(w)}_{E^nX^n}/\alpha_w$ with
  $\alpha_w = \tr`\big[P^{(w)}_{E^nX^n} \, \rho_{E^nX^n}]$.  Let
  \begin{align}
    \xi = \sum_{w>w'} \tr`\big[P^{(w)}_{E^nX^n}\rho_{E^nX^n}]
    = \sum_{w>w'} \alpha_w\ .
  \end{align}
  Then
  \begin{align}
    \hspace*{1em}&\hspace*{-1em}
    \mathcal{R}'_{E^nX^nW}`\Big(
      \rho_{E^nX^n}\otimes\tau_W^{(E_0)}
    )
                   \nonumber\\
    &=
      \sum_{w\in\mathscr{W}} \alpha_w \mathcal{R}^{w+\beta^{-1}\delta}_{E^nX^nW}`\Big(
        \rho^{(w)}_{E^nX^n} \otimes \tau_W^{(E_0)}
      )
                   \nonumber\\
    &=
      \sum_{w\leq w'} \alpha_w \mathcal{R}^{w+\beta^{-1}\delta}_{E^nX^nW}`\Big(
        \rho^{(w)}_{E^nX^n} \otimes \tau_W^{(E_0)}
      )
      \nonumber\\ &\qquad
      + \sum_{w>w'} \alpha_w \mathcal{R}^{w+\beta^{-1}\delta}_{E^nX^nW}`\Big(
        \rho^{(w)}_{E^nX^n} \otimes \tau_W^{(E_0)}
      )
                   \nonumber\\
    &=
      \sum_{w\leq w'} \alpha_w \mathcal{R}^{w+\beta^{-1}\delta}_{E^nX^nW}`\Big(
        \rho^{(w)}_{E^nX^n} \otimes \tau_W^{(E_0)}
      )
      \nonumber\\ &\qquad
      + \tilde\Delta_{E^nX^n}\ ,
     \label{eq:fekohry3809iojndkls}
  \end{align}
  defining $\tilde\Delta_{E^nX^n}$ as the entire second sum and noting that
  $\tr`(\tilde\Delta_{E^nX^n})\leq\xi$.
  The guarantees of \cref{thm:LemmaUnivErasureW0} ensure that the output of
  $\mathcal{R}^{w+\beta^{-1}\delta}_{E^nX^nW}$ in each term of the first sum
  produce a state close to $\tau_W^{(E'_{w+\beta^{-1}\delta})}$ on $W$ with
  $E_0 - E'_{w+\beta^{-1}\delta}$
  satisfying~\eqref{eq:LemmaUnivErasureW0WorkCost}.
  For $w''$ in the theorem claim and for any $w\leq w'$,
  \begin{align}
    w'' &= n(w' + 6\beta^{-1}\delta + F_E) + \log\poly(n)
          \nonumber\\
    &\geq n`\big[(w + \beta^{-1}\delta) + 5\beta^{-1}\delta + F_E] + \log\poly(n)
          \nonumber\\
    &\geq E_0 - E'_{w+\beta^{-1}\delta} + \log\poly(n)\ .
  \end{align}
  For all terms in the first sum in~\eqref{eq:fekohry3809iojndkls}, we find:
  \begin{align}
    \hspace*{1em}&\hspace*{-1em}
    \tr`\Big{ \Pi_W^{\geq(E_0 - w''+\log\poly(n))}
    \mathcal{R}_{E^nX^nW}^{(w+\beta^{-1}\delta)}`\Big(
    \rho^{(w)}_{E^nX^n} \otimes \tau_W^{(E_0)}
    ) }
    \nonumber\\
    &\geq
    \tr`\Big[ \Pi_W^{\geq(E_0 - w''+\log\poly(n))}
      \tau_W^{(E'_{w+\beta^{-1}\delta})} ]
    - \poly(n)\,\ee^{-n\eta/2}
    \nonumber\\
    &=
      1 - \poly(n)\,\ee^{-n\eta/2}\ ,
  \end{align}
  where the last equality holds because
  $E_0 - w'' + \log\poly(n) \leq E'_{W+\beta^{-1}\delta}$.
  Here, the $\log\poly(n)$ term is the work cost associated with resetting the
  measurement outcome of the $`{ P^{(w)}_{E^nX^n} }$'s, which is included in
  $\mathcal{R}_{E^nX^nW}$ but not in $\mathcal{R}'_{E^nX^nW}$.
  Combining the above, we find
  \begin{align}
    \hspace*{1em}&\hspace*{-1em}
    \tr`\Big[ \Pi_W^{\geq(E_0-w'')}
    \mathcal{R}_{E^nX^nW}`\big(\rho_{E^nX^n}\otimes\tau_W^{(E_0)} ) ]
                   \nonumber\\
    &\geq
      \sum_{w\leq w'} \alpha_w `\big(1 - \poly(n)\,\ee^{-n\eta/2})
      \nonumber\\
    &\geq (1 - \xi) - \poly(n)\,\ee^{-n\eta/2}\ ,
  \end{align}
  exactly as claimed.
  
  Finally, we
  prove~\ref{item:thmUnivCondErasureForEachIidInputWorkCostIidInput}.  Let
  $\rho_{EX}$ be any quantum state, and let
  \begin{align}
    W[\tr_E;\rho]
    &= \beta^{-1}`\big[
      \DD{\rho_X}{\Gamma_X} - \DD{\rho_{EX}}{\Gamma_{EX}} ]\,.
  \end{align}
  Recalling the definition of $P^{(w)}_{X^nE^n}$, we see that for any
  $w > W[\tr_E;\rho]+\beta^{-1}\delta$, we have that at least one of
  $E_k,E'_\ell,\beta^{-1}S(\lambda/n),\beta^{-1}S(\mu/n)$ must be
  $\beta^{-1}\delta/4$-far from $\tr(H_{XE}\rho_{XE})$, $\tr(H_X \rho_X)$,
  $\beta^{-1}S(\rho_{XE})$, or $\beta^{-1}S(\rho_X)$, respectively.  By the
  properties of the projectors $R_{E^nX^n}^{E_k}$, $\Pi_{E^nX^n}^\lambda$,
  $\Pi_{X^n}^\mu$, $S_{X^n}^{E'_\ell}$, we must have that there is some $\eta>0$
  such that for all $w > W[\tr_E;\rho]+\beta^{-1}\delta$,
  \begin{align}
    \tr`\big(  P^{(w)}_{E^nX^n} \rho_{EX}^{\otimes n} ) \leq \poly(n)\,e^{-n\eta}\ ;
    \label{eq:d09iodsajfknldsaj}
  \end{align}
  see for example Propositions~2.1 and~2.2 of Ref.~\cite{arXiv:1911.05563}.  For
  simplicity, we choose this $\eta$ to be equal to the one that was given to us
  by \cref{thm:LemmaUnivErasureW0}, simply because we can replace both $\eta$'s
  by the minimum of both values.   This now means that
  \begin{align}
    \sum_{w> W[\tr_E;\rho]+\beta^{-1}\delta}
    \tr`\Big[ P_{E^nX^n}^{(w)}\, \rho_{EX}^{\otimes n} ]
    \leq \poly(n)\,\ee^{-n\eta}\ .
  \end{align}
  We can now invoke
  point~\ref{item:thmUnivCondErasureForEachIidInputWorkCostAnyInput} with
  $\rho_{E^nX^n} \equiv \rho_{EX}^{\otimes n}$ and
  $w'=W[\tr_E;\rho] + \beta^{-1}\delta$ to see that
  \begin{multline}
    \tr`\Big[ \Pi_W^{\geq (E_0 - w'')}
      \mathcal{R}_{E^nX^nW}`\big( \rho_{EX}^{\otimes n} \otimes \tau^{(E_0)}_W ) ]
    \\[.5ex]
    \geq 1 - \poly(n)\,\ee^{-n\eta/2}\ ,
  \end{multline}
  with
  \begin{align}
    w'' = n`\big( W[\tr_E;\rho]+\beta^{-1}\delta + 6\beta^{-1}\delta + F_E ) + \log\poly(n)\ ,
  \end{align}
  completing the proof.
\end{proof}

It is left for us to prove the lemma.  This proof follows closely the proof
strategy of Theorem~7.1 of Ref.~\cite{Faist2021CMP_impl}.

\begin{proof}[*thm:LemmaUnivErasureW0]
  Let
  \begin{align}
    w_{\mathrm{max}}
    &= T(\tr_X)
      \nonumber\\
    &= \max_{\rho'} `\big[
    \beta^{-1} D(\rho'_X\Vert\Gamma_X)
    - \beta^{-1} D(\rho'_{EX}\Vert\Gamma_{EX})
    ].
  \end{align}
  If $w_0 \geq w_{\mathrm{max}} - 5\beta^{-1}\delta$, then we use the protocol based on
  thermal operations of Ref.~\cite{arXiv:1911.05563} in order to implement the
  map $\tr_X$ at a work cost asymptotically equal to $T(\tr_X)$, and satisfying
  all conditions of the lemma's claim.  For the rest of the proof, we suppose
  $w_0 < w_{\mathrm{max}} - 5\beta^{-1}\delta$.
  Furthermore, we note that
  (cf.\@ proof of Theorem 7.1 in \cite{arXiv:1911.05563}):
  \begin{align}
    \hspace*{1em}&\hspace*{-1em}
    w_0 + 5\beta^{-1}\delta
    < w_{\mathrm{max}}
      \nonumber\\
    &\leq
    \beta^{-1} D(\gamma_X\Vert\Gamma_X)
    - \beta^{-1} D(\gamma_{E}\otimes\gamma_X\Vert\Gamma_{E}\otimes\Gamma_X)
    \nonumber\\
    &= \beta^{-1} D(\gamma_E\Vert [\tr(\Gamma_E)]\gamma_E) = -F_E\ ,
    \label{eq:Lemmaw0UpBound}
  \end{align}
  where $F_E = -\beta^{-1}\log[\tr(\Gamma_E)]$.

  Let
  \begin{align}
    \hspace*{1em}&\hspace*{-1em}
    \mathscr{S}_{E^nX^n} =
    \nonumber\\
    & `\Big{
      \rho_{EX}^{\otimes n}\, :\:
    D(\rho_{EX}\Vert\Gamma_{EX})
    - D(\rho_X\Vert\Gamma_X)
    \geq - \beta w_0
    }
      \nonumber\\
    &\cup\ 
    `\Big{
    \hat\rho_{E^nX^n}\,:\:
    \tr`\big[P^{(\leq w_0)}_{E^nX^n} \hat\rho_{E^nX^n}] = 1
    }
      \ .
  \end{align}
  Let $P_{E^nX^n}^{-\beta w_0;\delta}$ be the universal typical and relative
  conditional operator furnished by Proposition~6.1 of
  Ref.~\cite{arXiv:1911.05563}, with respect to $\Gamma_{XE}$ and $\Gamma_X$.
  Specifically,
  \begin{align}
    P_{E^nX^n}^{-\beta w_0,\delta}
    &=
      \hspace*{-3em}
      \sum_{\substack{k,\ell,\lambda,\mu\,:\\
    E_k - \beta^{-1}S(\bar\lambda) - E_\ell + \beta^{-1}S(\bar\mu)
    \geq -\beta w_0-4\delta}}
      \hspace*{-3em}
    R^{E_\ell}_{X^n} \Pi^\mu_{X^n} \Pi^\lambda_{E^nX^n} R^{E_k}_{E^nX^n}
    \nonumber\\
    &= \sum_{\substack{w\in\mathscr{W}\,:\\
    w \leq w_0 + 4\beta^{-1}\delta}}
    P^{(w)}_{E^nX^n}
    \ .
  \end{align}
  Because $\Gamma_{XE}$ commutes with $\Ident_E\otimes\Gamma_X$, Proposition~6.1
  ensures that $P_{E^nX^n}^{-\beta w_0,\delta}$ is a Hermitian projector that
  commutes with $\Gamma_X^{\otimes n}$ and $\Gamma_{XE}^{\otimes n}$.
  As in the proof of Theorem~7.1 of Ref.~\cite{arXiv:1911.05563}, we show that
  $P_{E^nX^n}^{-\beta w_0,\delta}$ can perform a hypothesis test between any
  $\hat\rho_{E^nX^n}$ and $\gamma_E^{\otimes n}\otimes\hat\rho_{X^n}$ for any
  $\hat\rho_{E^nX^n}\in\mathscr{S}_{E^nX^n}$.  Definition~6.1 and
  Proposition~6.1 of Ref.~\cite{arXiv:1911.05563} ensure there exists $\eta>0$
  such that
  \begin{align}
    \tr`\big[ P_{E^nX^n}^{-\beta w_0;\delta} \rho_{EX}^{\otimes n} ]
    &\geq 1 - \kappa\ ;
    &
    \kappa &= \poly(n)\, e^{-n\eta}\ .
  \end{align}
  Then, consider any state $\hat\rho_{E^nX^n}$ such that
  $\tr`\big[P^{(\leq w_0)} \hat\rho_{E^nX^n}] = 1$.
  Then by construction of $P_{E^nX^n}^{-\beta w_0,\delta}$,
  \begin{align}
    \tr`\big[ P_{E^nX^n}^{-\beta w_0;\delta} \hat\rho_{E^nX^n} ]
    &= 1 \geq 1 - \poly(n)\,\ee^{-n\eta}\ .
  \end{align}

  We also evidently have
  $\Ident_X\otimes\Gamma_E = \Gamma_{X}^{-1/2}\Gamma_{XE}\Gamma_X^{-1/2}$, which
  combined with Point~(iii) of Definition~6.1 of Ref.~\cite{arXiv:1911.05563}
  gives us
  \begin{align}
    \hspace*{3em}&\hspace*{-3em}
    \tr_{E^n}`\big[ P_{E^nX^n}^{-\beta w_0;\delta} \Gamma_E^{\otimes n}]
    \nonumber\\
    &= `\big( \Gamma_X^{-1/2} )^{\otimes n}
    \tr_{E^n}`\big[ P_{E^nX^n}^{-\beta w_0;\delta} \Gamma_{XE}^{\otimes n} ]
      `\big( \Gamma_X^{-1/2} )^{\otimes n}
      \nonumber\\
    &\leq \poly(n)\,e^{-n(-\beta w_0-4\delta)}\,\Ident_{X^n}\ ,
  \end{align}
  further noting that $P_{X^nE^n}^{-\beta w_0;\delta}$ commutes with
  $\Gamma_{XE}^{\otimes n}$ and $\Gamma_X^{\otimes n}$.  We find, for any
  $\hat\rho_{X^n}$,
  \begin{align}
    \hspace*{3em}&\hspace*{-3em}
    \tr`\big[ P_{E^nX^n}^{-\beta w_0;\delta} \rho_{X^n}\otimes\gamma_E^{\otimes n}]
     \nonumber\\
    &\leq `\big[\tr`\big(\Gamma_E^{\otimes n})]^{-1}\,\poly(n)\,e^{-n(-\beta w_0-4\delta)}\,
      \tr(\rho_{X^n})
      \nonumber\\
    &= \poly(n)\,\exp`\big{ -n`\big(-\beta w_0 - 4\delta - \beta F_E) }\ .
  \end{align}
  Assume without loss of generality that $\eta \leq \delta$ (if $\eta > \delta$,
  we can replace $\eta$ by $\delta$ without affecting the proof).
  Using~\eqref{eq:Lemmaw0UpBound}, we see that
  $-\beta w_0-4\delta- \beta F_E- \eta > 0$.
  Now let
  \begin{align}
    m
    = \log\;\lfloor
    e^{n(-\beta w_0 - 4\delta -\beta F_E - \eta)}
    \rfloor \geq 0\ .
  \end{align}
  Then we have, for any $\hat\rho_{E^nX^n}\in\mathscr{S}_{E^nX^n}$, that
  \begin{align}
    \tr`\big[P_{E^nX^n}^{-\beta w_0;\delta} \hat\rho_{E^nX^n}]
    &\geq 1 - \kappa\ ;
    &
    \kappa &= \poly(n)\,e^{-n\eta}\ ;
    \nonumber\\
    \tr`\big[P_{E^nX^n}^{-\beta w_0;\delta} \,\gamma_E^{\otimes n}\otimes\hat\rho_{X^n}]
    &\leq e^{-m}\kappa'\ ;
    &
    \kappa' &= \poly(n)\,e^{-n\eta}\ .
  \end{align}
  Consider a quantum memory register $J$ of dimension $e^m$ (which is an
  integer), with trivial Hamiltonian $H_J = 0$ and initialized in a maximally
  mixed state; this corresponds to borrowing a heat bath with a trivial
  Hamiltonian in our thermal operation.  Let $\mathcal{R}^{(w_0)}_{E^nX^nJ}$ be
  the thermal operation provided by Proposition~7.1 of
  Ref.~\cite{arXiv:1911.05563} setting $S = E^n$, $M= X^n$, along with
  $\mathscr{S}_{E^nX^n}$, $P^{-\beta w_0;\delta}_{E^nX^n}$, $\gamma_E$, $m$,
  $\kappa$, $\kappa'$ as defined above.  We now view $J$ as an information
  battery that stores an amount $m' := -\beta^{-1} m$ amount of work, because
  $J$ is left in the pure state $\ket0_J$.  This Proposition~7.1 then guarantees
  that for any $\hat\rho_{E^nX^n} \in \mathscr{S}_{E^nX^n}$, the
  condition~\eqref{eq:LemmaUnivErasureW0Correct} is satisfied.  (It is
  straightforward to transfer the energy extracted in the information battery
  $J$ into the battery $W$, leaving the $J$ battery mixed and using it as an
  ancillary heat bath for our global thermal operation.)  Recall that
  \begin{align}
    m
    &= \log \lfloor e^{n(-\beta w_0 - 4\delta - \beta F_E - \eta) } \rfloor
      \nonumber\\
    &= n(-\beta w_0 - 4\delta - \beta F_E - \eta) - \nu\ ,
  \end{align}
  where $0\leq \nu \leq \log(2)$ accounts for the possible error caused by the
  floor function, noting that for any $y\geq 0$, we have
  $y - \log(2) \leq \log\lfloor e^y\rfloor \leq y$. Then (recall $\eta\leq\delta$):
  \begin{align}
    w' = -\beta^{-1} m \leq n (w_0 + F_E + 5\beta^{-1}\delta) + \beta^{-1}\log(2)\ ,
  \end{align}
  which completes the proof.
\end{proof}

\section{Proof of universal implementation of a work-cost-dephased channel with
  optimal per-i.i.d.\@ input work cost}
\label{appx:UnivImplPerIidInputProofs}

This section is devoted to the proof of \cref{thm:UnivImplPerIidInputCost},
using the result for conditional erasure derived in
\cref{appx:UnivImplPerIidInputCondErasureProofs}.

Any time-covariant channel $\mathcal{E}$ possesses a Stinespring dilation with
an energy-conserving unitary (e.g.\@ Lemma~7.2 in Ref.~\cite{arXiv:1911.05563},
cf.\@ also earlier works~\cite{Scutaru1979RMP_covariant,%
  Keyl1999JMP_cloning,PhDMarvian2012_symmetry}): There exists an environment
system $E$ with Hamiltonian $H_E$, as well as an energy-conserving unitary
$V_{XE}$ and a state $\ket0_E$ with $\dmatrixel{0}{H_E}_E = 0$ such that
$\mathcal{E}_{X\to X}(\cdot) = \tr_E`\big[ V `\big( \proj0_E\otimes(\cdot) )
V^\dagger ]$.

It turns out that in the time-covariant Stinespring picture, the
work-cost-dephasing operation coincides precisely with the dephasing operation
$\mathcal{D}^{\textup{W}}_{E^nX^n}$ in terms of the $P_{E^nX^n}^{(w)}$ operators
introduced in \cref{appx:UnivImplPerIidInputCondErasureProofs} for conditional
erasure.  In this picture, it becomes evident that
$\widetilde{\mathcal{E}}_{X^n}$ is indeed a valid quantum channel, namely, it is
completely positive and trace-preserving.

\begin{lemma}
  \label{thm:WorkCostDephasedChannelIsStinespringDephased}
  Let $\mathcal{E}_{X}$ be any time-covariant quantum channel and set $V_{XE}$
  be an energy-conserving Stinespring dilation unitary as above.
  Then its work-cost-dephased map satisfies
  \begin{multline}
    \widetilde{\mathcal{E}}_{X^n\to X^n}(\cdot)
    \\
    = \tr_{E^n}`\Big{
    \mathcal{D}^{\textup{W}}_{E^nX^n}`\Big(
      V_{XE}^{\otimes n} \, `\big[\proj 0_E^{\otimes n}\otimes`(\cdot)] \,
      (V_{XE}^{\otimes n})^\dagger
    )
    }\,.
  \label{eq:WorkCostDephasedChannelStinespringDephased}
  \end{multline}
\end{lemma}
\begin{proof}[*thm:WorkCostDephasedChannelIsStinespringDephased]
  The dephasing operations are indeed the same, which can be seen by mapping the
  operators applied on $R^n$ in~\eqref{eq:WorkCostMeasInputOutputPhat} to the
  output of the Stinespring dilation.  Specifically, we can see that for any
  $k,\ell,\lambda,\mu$,
  \begin{align}
    \hspace*{1em}&\hspace*{-1em}
   `\big[
      `(S_{X^n}^{E_\ell}\Pi^\mu_{X^n})\otimes`(S_{X^n}^{E_k}\Pi^\lambda_{X^n})_{R^n}^t
    ] \;
    V_{XE}^{\otimes n}\,`(\ket0_E\otimes \ket\Phi_{X:R})^{\otimes n}
    \nonumber\\[.5ex]
    &= 
      `(S_{X^n}^{E_\ell}\Pi^\mu_{X^n})
      \; V_{XE}^{\otimes n} \;
      `(S_{X^n}^{E_k} \otimes \Pi^\lambda_{R^n}) \,
      `(\ket0_E\otimes \ket\Phi_{X:R})^{\otimes n}
    \nonumber\\[.5ex]
    &=
   `\big[
      S_{X^n}^{E_\ell}\Pi^\mu_{X^n} \, \Pi^\lambda_{(EX)^n} R_{E^nX^n}^{E_k}
    ] \;
    V_{XE}^{\otimes n}\,`(\ket0_E\otimes \ket\Phi_{X:R})^{\otimes n}\ ,
  \end{align}
  where $R_{E^nX^n}^{E_k}$ is defined as in the definition of $P^{(w)}_ {E^nX^n}$
  in~\eqref{eq:WorkCostMeasEnvXpP}.  
  In the second inequality above, we first move $S^{E_k}_{X^n}$ past $V_{XE}$ as
  $R^{E_k}_{E^nX^n}$, since eigenspaces of $H_{X^n} = \sum_{i=1}^n H_{X_i}$ are
  mapped onto eigenspaces of $H_{E^nX^n} = \sum_{i=1}^n (H_{X_i} + H_{E_i})$; we
  then move the Schur-Weyl block projector $\Pi_{R^n}^\lambda$ over past the
  unitary $V_{XE}^{\otimes n}$ as $\Pi_{E^nX^n}^\lambda$, which is justified
  from e.g.\@ Proposition~C.1 of Ref.~\cite{Faist2025arXiv_maximum} (cf.\@ also
  Ref.~\cite{PhDWalter2014}) given that
  $(V_{XE}\ket0_E\ket{\Phi_{X:R}})^{\otimes n}$ lives in the symmetric subspace
  of $E^nX^nR^n$.
  We write
  \begin{align}
    A_{E^nX^n} = V_{XE}^{\otimes n} `[\proj0_E^{\otimes n}\otimes
    `(\Phi_{X^n:R^n})] (V_{XE}^{\otimes n})^\dagger\ ,
  \end{align}
  so that \cref{eq:WorkCostDephasedMap} becomes
  \begin{align}
      \widetilde{\mathcal{E}}_{X^n}(\Phi_{X^n:R^n})
      &= \sum_w \hat{P}^{(w)}_{X^nR^n}
        \mathcal{E}_{X^n}(\Phi_{X^n:R^n})
        \hat{P}^{(w)}_{X^nR^n}
        \nonumber\\
      &= \sum_w \hat{P}^{(w)}_{X^nR^n}
        \tr_{E^n}`(  A_{E^nX^n} )
        \hat{P}^{(w)}_{X^nR^n}
        \nonumber\\
      &= \sum_w
        \tr_{E^n}`\Big{  P^{(w)}_{E^nX^n}
          A_{E^nX^n}
          P^{(w)}_{E^nX^n}
        }\,,
  \end{align}
  which is the claimed expression for $\widetilde{\mathcal{E}}_{X^n}$.
\end{proof}

\begin{corollary}
  \label{thm:WorkCostDephasedMapIsCpTp}
  Let $\widetilde{\mathcal{E}}_{X^n\to X^n}$ be the work-cost-dephased map
  associated with some arbitrary quantum channel $\mathcal{E}_{X^n}$.  Then the
  map $\widetilde{\mathcal{E}}_{X^n\to X^n}$ is completely positive and trace
  preserving.
\end{corollary}
\begin{proof}[**thm:WorkCostDephasedMapIsCpTp]
  Immediate from the
  expression~\eqref{eq:WorkCostDephasedChannelStinespringDephased} for
  $\widetilde{\mathcal{E}}_{X^n}$, which is a composition of cp.\@ tp.\@ maps.
\end{proof}

\begin{proof}[*thm:UnivImplPerIidInputCost]
  Let $E$, $H_E$, $V_{XE}$ as above.
  We define for later convenience
  \begin{align}
  \Gamma_X
  &= e^{-\beta H_X}\ ;
  \qquad\quad
  \Gamma_E
    = e^{-\beta H_E}\ ;
  \nonumber\\
    \Gamma_{XE}
    & = e^{-\beta(H_X + H_E)} = \Gamma_X\otimes\Gamma_E\ .
  \end{align}
  We build the desired thermal operation as follows.  Overall, the protocol
  attempts to implement $\mathcal{E}^{\otimes n}$ by explicitly implementing its
  Stinespring dilation.  Specifically, let $\mathcal{T}_{X^nW}$ be the thermal
  operation which performs the following steps:
  \begin{enumerate}[label={\arabic*.}]
  \item Bring in ancillary systems $E^n$, each with a Hamiltonian
    $H_{E_i} = H_E$ as required to host an energy-covariant Stinespring dilation
    of $\mathcal{E}_{X}$.  The $E^n$ systems are initialized in their thermal
    state, so that this operation can be performed for free in a thermal
    operation;
  \item Reset all the $E^n$ systems to their state $\ket0_{E}$, by drawing the
    necessary work from the battery $W$.  This erasure costs $-nF_E$ work, also
    in a single instance of the
    process~\cite{Horodecki2013_ThermoMaj,Aberg2013_worklike}, so the battery
    must transition from the initial charge state $\tau_W^{(E_0)}$ to the state
    $\tau_W^{(E_0 + nF_E)}$.  (This value is also consistent with the usual
    Landauer formula if $E$ has a trivial Hamiltonian; in this case,
    $\Gamma_E=\Ident_E$ and $F_E = -\beta^{-1}\log(\dim(E))$.)
  \item Apply the energy-conserving unitary $V_{XE}^{\otimes n}$.
    Energy-conserving unitaries are free, so this operation costs no work.
  \item Invoke the conditional erasure protocol given by
    \cref{thm:UnivCondErasureForEachIidInput} to extract as much work as
    possible from the $E^n$ systems, conditioned on the $X^n$ systems, and
    leaving the $E^n$ in their thermal state.  Let $\mathcal{R}_{E^nX^nW}$ be
    the thermal operation provided by \cref{thm:UnivCondErasureForEachIidInput}.
    This operation consumes a variable amount of work from the battery $W$.  The
    resulting state on $X^n$ contains the dephasing operation
    $\mathcal{D}^{\textup{W}}_{E^nX^n}$, which according to
    \cref{thm:WorkCostDephasedChannelIsStinespringDephased} precisely consists
    in implementing $\widetilde{\mathcal{E}}_{X^n}$.
  \item Discard the $E^n$ systems.
  \end{enumerate}
  
  The desired fidelity bound appears naturally.  If the input state is
  $\ket\sigma_{X^nR}$, then the state immediately after applying
  $V_{XE}^{\otimes n}$ is exactly
  $V_{XE}^{\otimes n} `(\proj0_E^{\otimes n}\otimes\sigma_{X^nR})
  (V_{XE}^\dagger)^{\otimes n}$.  With
  \cref{thm:WorkCostDephasedChannelIsStinespringDephased}:
  \begin{align}
    \hspace*{1em}&\hspace{-1em}
    F`\Big(
      \tr_W`\big[ \mathcal{T}_{X^nW}( \sigma_{X^nR} \otimes \tau_W^{(E_0)} ) ]
      \,,\;
      \widetilde{\mathcal{E}}_{X^n}`\big[ \sigma_{X^nR} ]
    )
    \nonumber\\[.5ex]
    &= F \Bigl(
      \tr_W`\big[ \mathcal{T}_{X^nW}( \sigma_{X^nR} \otimes \tau_W^{(E_0)} ) ]
      \,,\;
      \nonumber\\[.5ex]&\qquad
      \tr_{E^n}`\big{
      \mathcal{D}_{E^nX^n}^{(w)}`\big[
        V_{XE}^{\otimes n} (\proj0_E^{\otimes n}
        \otimes \sigma_{X^nR}) (V_{XE}^{\otimes n})^\dagger
      ]
    }
    \Bigr)
    \nonumber\\[.5ex]
    &\geq 1 - \poly(n)\,\ee^{-n\eta}\ ,
      \label{eq:dapiewqfbdjilabsfkddieuwio}
  \end{align}
  where the last inequality is provided by
  point~\ref{item:thmUnivCondErasureForEachIidInputOutputFidelity} of
  \cref{thm:UnivCondErasureForEachIidInput}.
  The ``furthemore'' part of the
  claim~\ref{item:thmUnivImplPerIidInputCostFidelityBound} follows from
  \cref{thm:WrkCostDephasProcessPINoRefStateTransf}.
  This proves~\ref{item:thmUnivImplPerIidInputCostFidelityBound}.

  From the way we set up our protocol it is clear that the final battery charge
  is variable and is heavily dependent on the conditional erasure step.  If the
  input state is $\sigma_X^{\otimes n}$, then immediately after applying the
  unitary, we have the state $\rho_{XE}^{\otimes n}$, where
  $\rho_{XE} = V_{XE} `(\proj0_E\otimes\sigma_X) V_{XE}^\dagger$.
  Point~\ref{item:thmUnivCondErasureForEachIidInputWorkCostIidInput} of
  \cref{thm:UnivCondErasureForEachIidInput} then asserts that, with high
  probability, the conditional erasure step does not deplete the battery by any
  amount greater than
  \begin{align}
    w_\rho
    &= n \bigl[ \beta^{-1}`\big( \DD{\rho_X}{\Gamma_X} - \DD{\rho_{EX}}{\Gamma_{EX}} )
      \nonumber\\[.5ex]&\qquad
      + F_E + 7\beta^{-1}\delta \bigr] + \log \poly(n)
    \nonumber\\
    &= n \bigl[ \beta^{-1}`\big( \DD{\mathcal{E}(\sigma_X)}{\Gamma_X}
      - \DD{\sigma_{X}}{\Gamma_{X}} )
      \nonumber\\[.5ex]&\qquad
      + F_E + 7\beta^{-1}\delta \bigr] + \log \poly(n)
    \ .
  \end{align}
  Accounting for the initial work cost $- n F_E$ of resetting the $E^n$
  registers to $\ket{0}_E^{\otimes n}$ which
  proves~\ref{item:thmUnivImplPerIidInputCostVariableWorkCostBound}.
\end{proof}

\section{Universal, variable-work-cost implementation of an arbitrary channel
  using Gibbs-preserving maps, with work-cost process dephasing}
\label{appx:UnivVariableImplWithGPM}

We prove the theorem, which has as an immediate consequence the existence of a
thermodynamic implementation of the work-cost-dephased map associated with any
arbitrary i.i.d.\@ channel $\mathcal{E}_{X\to X'}^{\otimes n}$, with optimal
variable work cost for each i.i.d.\@ input state.  Here, the output system $X'$
and Hamiltonian $H_{X'}$ may differ from $X$ and $H_{X}$.  As earlier, we fix
$\beta>0$ and set $\Gamma_X = \ee^{-\beta H_X}$,
$\Gamma_{X'} = \ee^{-\beta H_{X'}}$.

Let $\mathcal{E}_{X\to X'}$ be any arbitrary quantum channel.  Let
$V_{X\to EX'}$ be a Stinespring dilation isometry, with
$\mathcal{E}_{X\to X'}(\cdot) = \tr_E`\big{ V_{X\to EX'}\,`(\cdot)\,V^\dagger }$.
We define
\begin{align}
  \Gamma_{EX'} = V_{X\to EX'} \, \Gamma_X\,  V^\dagger\ .
\end{align}

Let $`{ S_{X'^n}^{E'_\ell} }$ and $`{ R_{E^nX'^n}^{E_k} }$ be the projectors onto
the eigenspaces of $\Gamma_{X'}^{\otimes n}$ and $\Gamma_{EX'}^{\otimes n}$,
respectively; let $\Pi^\mu_{X'^n}$ and $\Pi^\lambda_{(EX')^n}$ be the projectors
onto the Schur-Weyl block of $X'^n$ and $(EX')^n$ labeled by Young diagram $\mu$
and $\lambda$, respectively.  Analogously to~\eqref{eq:WorkCostMeasEnvXpP}, and
following Proposition~6.1 of Ref.~\cite{Faist2021CMP_impl}, we define
\begin{align}
  M^{(w)}_{E^nX'^n}
  :=
  \hspace*{-2em}
  \sum_{\substack{k,\ell,\lambda,\mu:\\
  E_k - \beta^{-1}S(\lambda/n) - E'_\ell + \beta^{-1}S(\mu/n) = \beta w
  }}
  \hspace*{-3em}
  S_{X'^n}^{E'_\ell} \Pi_{X'}^\mu \Pi_{(EX')^n}^\lambda R_{E^nX'^n}^{E_k}\ .
\end{align}
As noted in Proposition~6.1 of Ref.~\cite{Faist2021CMP_impl} and its proof, the
operator $M^{(w)}_{E^nX'^n}$ is generically not a projector; it might even not
be Hermitian.  It always satisfies
$M^{(w)\,\dagger}_{E^nX'^n} M^{(w)}_{E^nX'^n} \leq\Ident$.  Moreover,
$`\big{ M^{(w)\,\dagger}_{E^nX'^n} M^{(w)}_{E^nX'^n} }_w$ forms a POVM.
Here again, the possible values of $w$ attained by combinations of the
$\poly(n)$ many $k,\ell,\lambda,\mu$ forms a set $\mathscr{W}$ of size
$\abs{\mathscr{W}}\leq\poly(n)$.

For this construction based on Gibbs-preserving maps, the work-cost-dephased
version of $\mathcal{E}^{\otimes n}$ takes the form
\begin{multline}
  \widetilde{\mathcal{E}}_{X^n\to X'^n}(\cdot) = 
  \\[1ex]
  \sum_w \tr_{E^n}`\Big{
  M^{(w)}_{E^nX'^n}\,V_{X\to EX'}^{\otimes n}\,`(\cdot) \,
  V^{\otimes n\,\dagger}\,M^{(w)\,\dagger}_{E^nX'^n}
  }\,.
\end{multline}

The construction follows:

\begin{theorem}[Universal, variable-work-cost implementation of an arbitrary i.i.d.\@ channel
  with work-cost dephasing, using Gibbs-preserving maps]
  \label{thm:UnivVarWorkIidGPM}
  Let $\mathcal{E}_{X\to X'}$ be any arbitrary quantum process and let $W$ be a
  battery system as above, with $\Gamma_W=\Ident_W$.  There exists a completely
  positive, trace-nonincreasing map $\mathcal{T}_{X^nW\to X^nW}$ such that:
  \begin{enumerate}[label=(\roman*)]
  \item We have
    \begin{align}
      \mathcal{T}_{X^nW\to X^nW}`\big(\Gamma_X^{\otimes n}\otimes\Gamma_W)
      \leq \Gamma_X^{\otimes n}\otimes\Gamma_W\ ;
    \end{align}
  \item For any arbitrary state $\sigma_{X^nR}$, 
    \begin{multline}
      \tr_W`\big[\mathcal{T}_{X^n\to X'^n}(\sigma_{X^nR}\otimes\tau_W^{(E_0)})]
      = \widetilde{\mathcal{E}}_{X^n\to X'^n}(\sigma_{X^nR}) \ ;
    \end{multline}
  \item For any $\sigma_X$, we have
    \begin{multline}
      \tr`\Big[ \Pi_W^{\geq(E_0 - w_\sigma)} 
      \mathcal{T}_{X^n\to X'^n}(\sigma_{X}^{\otimes n}\otimes\tau_W^{(E_0)})
      ]
      \\[.5ex]
      \geq 1 - \poly(n)\,\ee^{-n\eta}\ ,
    \end{multline}
    where
    \begin{align}
      w_\sigma
      &=  n `\Big( W[\mathcal{E};\sigma] + \beta^{-1}\delta) + \log\poly(n)\ .
    \end{align}
  \end{enumerate}
\end{theorem}

\begin{proof}[**thm:UnivVarWorkIidGPM]
  Define
  \begin{align}
    \mathcal{T}^{(w)}_{X^n\to X'^n}(\cdot)
    = M^{(w)}_{E^nX'^n}\,V_{X\to EX'}^{\otimes n}\,(\cdot)\,
    V^{\otimes n\,\dagger}\,M^{(w)\,\dagger}_{E^nX'^n}\ ,
  \end{align}
  and
  \begin{align}
    \Delta^{(w)}_{W\to W}(\cdot)
    = \tr_W`\Big[\Pi_W^{\geq E_0}\,(\cdot)]\,\tau_W^{(E_0 - w)}\ .
  \end{align}
  Then let
  \begin{align}
    \mathcal{T}_{X^nW\to X'^nW}(\cdot)
    = \sum_{w\in\mathscr{W}}
    `\big[\mathcal{T}^{(w)} \otimes \Delta^{(w+\poly(n))}_W](\cdot)\ .
  \end{align}
  We find
  \begin{align}
    \mathcal{T}^{(w)}(\Gamma_X^{\otimes n})
    &\leq \tr_{E^n}`\big[ M^{(w)}_{E^nX'^n}\,\Gamma_{EX'}^{\otimes n}
      M^{(w)\,\dagger}_{E^nX'^n} ]
      \nonumber\\
    &\leq \poly(n)\, \ee^{n\beta w}\, \Gamma_{X'}^{\otimes n}\ ,
  \end{align}
  as we can see from the proof of Proposition~6.1 in
  Ref.~\cite{Faist2021CMP_impl}.
  We then see that
  \begin{align}
    \hspace*{2em}&\hspace*{-2em}
    \mathcal{T}_{X^nW\to X'^nW}(\Gamma_X^{\otimes n}\otimes\Gamma_W)
    \nonumber\\[.5ex]
    &\leq \sum_w
    \mathcal{T}^{(w)}(\Gamma_X^{\otimes n})\otimes\ee^{-\beta(w+\poly(n))}\Ident_W
    \nonumber\\
    &\leq 
      \Gamma_{X'}^{\otimes n}\otimes\Ident_W\ ,
  \end{align}
  as desired.

  For any quantum state $\sigma_{X^nR}$, we find that
  \begin{align}
    \tr_W`\Big[ \mathcal{T}`\big(\sigma_{X^nR}\otimes\tau_W^{(E_0)}) ]
    &= \sum_w \mathcal{T}^{(w)}`\big(\sigma_{X^nR})
      \nonumber\\
    &= \widetilde{\mathcal{E}}(\sigma_{X^nR})\ .
  \end{align}

  Finally, let $\sigma_X$ be any quantum state on $X$.  We seek a property
  analogous to~\eqref{eq:d09iodsajfknldsaj}.  Whenever $w$ is at least
  $\beta^{-1}\delta$-far from $W[\mathcal{E};\sigma]$, we have that at least one
  of $\tr`(\sigma_X H_X)$, $\beta^{-1}S(\sigma_X)$,
  $\tr`(\mathcal{E}(\sigma) H_{X'})$, $\beta^{-1}S(\mathcal{E}(\sigma))$ must be
  at least $\beta^{-1}\delta/4$-far away from $E_k$, $\beta^{-1}S(\lambda/n)$,
  $E'_\ell$, $\beta^{-1}S(\mu/n)$, respectively (cf.\@ proof of Proposition~6.1
  of Ref.~\cite{Faist2021CMP_impl}, as well as Propositions~2.1 and~2.2 of that
  reference).  It follows that for any such $w$, we have
  \begin{multline}
    \tr`\big( M_{E^nX'^n}^{(w)\,\dagger} M_{E^nX'^n}^{(w)} \;
    V_{X\to EX'}^{\otimes n} \sigma_X^{\otimes n} (V^{\otimes n})^\dagger)
    \\[.5ex]
    \leq \poly(n)\,\ee^{-n\eta}\ .
  \end{multline}
  Therefore, up to probability at most $\poly(n)\,\ee^{-n\eta}$, only work
  events $w$ contribute where $w\approx W[\mathcal{E};\sigma]$.  This proves the
  final property.
\end{proof}

\section{Effect of work-cost dephasing}
\label{appx:WorkCostDephasingDetails}

In this appendix we derive the proofs of the statements made in
\cref{sec:WorkCostCovariance}.

We prove \cref{thm:WrkCostDephasBipartPINoRef} of the main text, which states
that the work-cost-dephasing $\mathcal{D}^{(\textup{W})}_{E^nX^n}$
[cf.~\cref{eq:WorkCostDephasingMapEX}] on a time-covariant,
permutation-invariant state $\rho_{E^nX^n}$ has no effect if $E^n$ is traced
out.
\begin{proof}[*thm:WrkCostDephasBipartPINoRef]
  Observe that all four projectors in~\eqref{eq:WorkCostDephasingMapEX} commute.
  Then, because $\rho_{E^nX^n}$ is permutation-invariant, for any
  $\lambda,\lambda',\mu,\mu'$,
  \begin{align}
    \hspace*{2em}&\hspace*{-2em}
    \Pi_{X^n}^\mu \tr_{E^n}`\big[
      \Pi_{(EX)^n}^{\lambda}
      \rho_{E^nX^n}
      \Pi_{(EX)^n}^{\lambda'}
    ]     \Pi_{X^n}^{\mu'}
    \\
    &= \delta_{\lambda,\lambda'} \,
    \Pi_{X^n}^\mu
    \tr_{E^n}`\big[
      \Pi_{(EX)^n}^{\lambda}
      \rho_{E^nX^n}
    ]
    \Pi_{X^n}^{\mu'}
    \nonumber\\
    &= \delta_{\lambda,\lambda'} \delta_{\mu,\mu'} \,
    \Pi_{X^n}^\mu
    \tr_{E^n}`\big[
      \Pi_{(EX)^n}^{\lambda}
      \rho_{E^nX^n}
    ]
    \ ,
    \label{eq:dfabuibfhreu9qfjk}
  \end{align}
  where the last step follows from the fact that the expression
  $\tr_{E^n}(\cdot)$ is permutation-invariant over the copies $X^n$.
  Furthermore,
  $R_{E^nX^n}^{E_k} \rho_{E^nX^n} R_{E^nX^n}^{E_{k'}} = \delta_{k,k'}\,
  R_{E^nX^n}^{E_k} \rho_{E^nX^n} R_{E^nX^n}^{E_k}$ because $\rho_{E^nX^n}$ is
  time-covariant.  We can then write
  \begin{align}
    \hspace*{1em}&\hspace*{-1em}
    \tr_{E^n}`\Big{
     \mathcal{D}^{(\textup{W})}_{E^nX^n} `\big(\rho_{E^nX^n})
    }
    = \sum_{w \in\mathscr{W}}
    P^{(w)}_{E^nX^n} \rho_{E^nX^n} P^{(w)}_{E^nX^n}
    \nonumber\\
    &=
      \hspace*{-2.8em}
      \sum_{\substack{
      w,\lambda,\mu,k,\ell:\\
      E_k-\beta^{-1}S(\bar\lambda)\ \ \  \\\ \ \ \ \ \ -E_\ell+\beta^{-1}S(\bar\mu) = w
    }} \hspace*{-2.8em}
      S_{X^n}^{E_\ell} \Pi^{\mu}_{X^n} \tr_{E^n}`\Big{
    \Pi^{\lambda}_{(EX)^n}
    R_{E^nX^n}^{E_k}
    \rho_{E^nX^n}
    } S_{X^n}^{E_{\ell}}\,,
    \nonumber\\
    &=
      \tr_{E^n}`\big{
      \rho_{E^nX^n}
      } = \rho_{X^n}\ ,
  \end{align}
  where we can use the same $\ell$ on both $S_{X^n}^{E_{\ell}}$ terms because
  the $\ell$ inside each $P^{(w)}_{E^nX^n}$ operator is fixed from the
  constraint
  $E_\ell = E_k - \beta^{-1}S(\bar\lambda) + \beta^{-1}S(\bar\mu) - w$, and
  where the final equality follows from the facts that the trace term
  $\tr_{E^n}(\cdot)$ is time-covariant (making the second $S_{X^n}^{E_\ell}$
  redundant) and that the sum simply ranges over all possible
  $k,\ell,\lambda,\mu$.
\end{proof}

A corresponding claim holds for implementing quantum channels.
This is, in fact, an immediate consequence of
\cref{thm:WrkCostDephasBipartPINoRef} and
\cref{thm:WorkCostDephasedChannelIsStinespringDephased}:
\begin{proof}[*thm:WrkCostDephasProcessPINoRefStateTransf]
  Let $V_{XE}$ be a time-covariant Stinespring dilation of $\mathcal{E}$ as in
  \cref{thm:WorkCostDephasedChannelIsStinespringDephased}.  Setting
  \begin{align}
    \rho_{E^nX^n} = 
    V_{XE}^{\otimes n} \, `\big[\proj 0_E^{\otimes n}\otimes\sigma_{X^n}] \,
    (V_{XE}^{\otimes n})^\dagger\ ,
  \end{align}
  we find by definition
  \begin{align}
    \tr_{E^n}(\rho_{E^nX^n})
    = \mathcal{E}^{\otimes n}(\sigma_{X^n})\ .
  \end{align}
  The claim follows by \cref{thm:WrkCostDephasBipartPINoRef}, since
  $\rho_{E^nX^n}$ is time covariant and permutation invariant.
\end{proof}

For the rest of this appendix, we assume that all systems have trivial
Hamiltonians: $H_E=0$, $H_X=0$.
We would like to show that the conditional erasure channel on $EX\to EX$,
\begin{align}
  \mathcal{E}_{\mathrm{cond.reset}\,EX}`(\cdot)
  &= (\tr_E\otimes\IdentProc[X]{})`(\cdot)\ ,
\end{align}
is not work-cost covariant.
Namely, we exhibit an example of a pure input state $\ket\rho_{E^nX^nR}$ such
that
\begin{align}
  \mathcal{E}_{\mathrm{cond.reset}\,EX}^{\otimes n}`(\rho_{E^nX^nR})
  \not\approx
  \widetilde{\mathcal{E}}_{\mathrm{cond.reset}\,EX}`(\rho_{E^nX^nR})\ ,
  \label{eq:WorkCostDephasCondResetNotWkCovar}
\end{align}
where $\widetilde{\mathcal{E}}_{\mathrm{cond.reset}\,EX}$ is the
work-cost-dephased channel associated with
$\mathcal{E}_{\mathrm{cond.reset}\,EX}^{\otimes n}$.

This argument provides an example of a state $\rho_{E^nX^nR}$ such that
$\tilde\rho_{X^nR} \neq \rho_{X^nR}$, where
$\rho_{E^nX^nR} = \mathcal{D}^{(\textup{W})}_{E^nX^n}(\rho_{E^nX^nR})$.  Given
that $H_X=H_E=0$, the state is trivially time covariant.  This is a situation in
which a work-cost dephasing affects the output state in
\cref{thm:UnivCondErasureForEachIidInput}.

The example state is constructed as follows.  The state on $E^nX^n$ is a convex
combination of a constant number of mixed i.i.d.\@ states, with fixed
probabilities $p_i$ independent of $n$:
\begin{align}
  \rho_{E^nX^n} = \sum p_i \, \sigma_{EX}^{(i)\,\otimes n}\ .
\end{align}
Such a state can be purified onto a system $R \equiv \bar{R}^n R'$, given
purifications $\ket[\big]{\sigma^{(i)}}_{EX\bar R}$ of each $\sigma_{EX}^{(i)}$, as
\begin{align}
  \ket\rho_{E^nX^n\bar{R}^nR'}
  = \sum \sqrt{p_i} \, \ket[\big]{\sigma^{(i)}}_{EX\bar R}^{\otimes n}\otimes\ket{i}_{R'}\ .
\end{align}
Then we have
\begin{multline}
  \rho_{X^nR} = \rho_{X^n\bar{R}^nR'}
  \\
  = \sum_{i,j} \sqrt{p_i p_j}
  \tr_{E^n}`\Big{ \ketbra[\big]{\sigma^{(i)}}{\sigma^{(j)}}_{EX\bar R}^{\otimes n} }
  \otimes\ketbra{i}{j}_{R'}\ .
  \label{eq:fokdih0q9wiadokfjdi}
\end{multline}
Our task is to show that this state differs significantly from
$\tr_{E^n}`\big{ \mathcal{D}^{(\textup{W})}(\rho_{E^nX^nR})}$, for suitable
choices of the states $`{ \ket{\sigma^{(i)}}_{EX\bar R} }$.
\begin{widetext}
We investigate the dephased state
$\tilde\rho_{E^nX^n\bar{R}^nR'} =
\mathcal{D}^{(\textup{W})}(\rho_{E^nX^n\bar{R}^nR'})$.  To simplify the
expression, we consider the operators $\matrixel{i}{\cdot}{j}_{R'}$
individually:
\begin{align}
  \frac1{\sqrt{p_ip_j}} \bra{i}_{R'}\, \tilde\rho_{E^nX^n\bar{R}^nR'}\,\ket{j}_{R'}
  &= 
    \sum_{w\in\mathscr{W}}
    \sum_{\substack{\lambda,\mu,\lambda',\mu':\\
      S(\bar\mu)-S(\bar\lambda)=\beta w\\
      S(\bar\mu')-S(\bar\lambda')=\beta w
  }}
  \Pi_{X^n}^\mu \Pi_{(EX)^n}^\lambda \,
    \ketbra[\big]{\sigma^{(i)}}{\sigma^{(j)}}_{EX\bar R}^{\otimes n} \,
  \Pi_{(EX)^n}^{\lambda'} \Pi_{X^n}^{\mu'}\ .
  \label{eq:afijegqw9dfkfdnq9id}
\end{align}
Using Eq.~(6.23) of Ref.~\cite{PhDHarrow2005},
\begin{subequations}
\begin{align}
  \norm[\big]{ \Pi_{(EX)^n}^\lambda \ket{\sigma^{(i)}}^{\otimes n} }^2
  &= \tr`\big[ \Pi_{(EX)^n}^\lambda \sigma_{EX\bar R}^{(i)\,\otimes n} ]
  \leq \poly(n)\,\exp`\Big{
    -\frac{n}{2}\,\onenorm[\big]{\bar\lambda - \spec`\big(\sigma^{(i)}_{EX})}^2
  }\ ;
    \\
  \norm[\big]{ \Pi_{X^n}^\mu \ket{\sigma^{(i)}}^{\otimes n} }^2
  &= \tr`\big[ \Pi_{X^n}^\mu \sigma_{EX\bar R}^{(i)\,\otimes n} ]
  \leq \poly(n)\,\exp`\Big{
    -\frac{n}{2}\,\onenorm[\big]{\bar\mu - \spec`\big(\sigma^{(i)}_{X})}^2
  }\ .
    \label{eq:ewu90ofjknlsadosakl}
\end{align}
\end{subequations}
Therefore, the object $\Pi_{X^n}^\mu \Pi_{(EX)^n}^\lambda %
\, \ket[\big]{\sigma^{(i)}}_{EX\bar R}^{\otimes n}$ has exponentially small norm
unless
$\onenorm[\big]{ \bar\lambda - \spec`\big(\sigma_{EX}^{(i)}) } \leq O(\sqrt{n})$
and $\onenorm[\big]{ \bar\mu - \spec`\big(\sigma_{X}^{(i)}) } \leq O(\sqrt{n})$.
Then, for any $\delta'>0$, there is an $\eta>0$ as well as an operator
$\Delta_{E^nX^n\bar{R}^n}$ with $\onenorm{\Delta} \leq \poly(n)\exp`{-n\eta}$,
such that
\begin{align}
  \eqref{eq:afijegqw9dfkfdnq9id} - \Delta
  &= 
    \sum_{w\in\mathscr{W}}
    \sum_{\substack{\lambda,\mu,\lambda',\mu':\\
      S(\bar\mu)-S(\bar\lambda)=\beta w\\
      S(\bar\mu')-S(\bar\lambda')=\beta w\\
     \onenorm{\bar\lambda - \spec(\sigma^{(i)}_{EX})} \leq \delta' \\
     \onenorm{\bar\lambda' - \spec(\sigma^{(j)}_{EX})} \leq \delta' \\
     \onenorm{\bar\mu - \spec(\sigma^{(i)}_{X})} \leq \delta'\\
     \onenorm{\bar\mu' - \spec(\sigma^{(j)}_{X})} \leq \delta'
  }}
  \Pi_{X^n}^\mu \Pi_{(EX)^n}^\lambda \,
    \ketbra[\big]{\sigma^{(i)}}{\sigma^{(j)}}_{EX\bar R}^{\otimes n} \,
  \Pi_{(EX)^n}^{\lambda'} \Pi_{X^n}^{\mu'}\ .
  \label{eq:fdkjhar8932ioawdsjk}
\end{align}

We now make the following assumption on the $\sigma^{(i)}_{EX\bar R}$:  There
exists $\delta>0$ such that
\begin{align}
  \spec`\big(\sigma_{EX}^{(i)}) & = \spec`\big(\sigma_{EX}^{(j)}) \quad\forall\ i,j\ ;
  &&\text{and}
  &
    \abs[\Big]{ S`\big(\sigma_X^{(i)}) - S`\big(\sigma_X^{(j)}) }
  &\geq \delta\quad\forall\ i\neq j\ .
    \label{eq:f8eiujdskljdaofi9io}
\end{align}

Let us investigate objects of the type
$\Pi_{X^n}^\mu \text{\eqref{eq:afijegqw9dfkfdnq9id}}\,\Pi_{X^n}^{\mu'}$, for
$\mu,\mu'$ corresponding to certain ``far away off-diagonal blocks.''  For
$i\neq j$, suppose that $\mu,\mu'$ are such that
\begin{align}
  \onenorm[\big]{\bar\mu - \spec`\big(\sigma_{X}^{(i)})}
  &\leq \delta'\ ;
  &
  \onenorm[\big]{\bar\mu' - \spec`\big(\sigma_{X}^{(j)})}
  &\leq \delta'\ .
    \label{eq:fidow9hewqonkj}
\end{align}
Let $\lambda,\lambda'$ be any Young diagrams such that
$\onenorm[\big]{\bar\lambda - \spec`\big(\sigma_{EX}^{(i)})}\leq\delta'$ and
$\onenorm[\big]{\bar\lambda' - \spec`\big(\sigma_{EX}^{(j)})}\leq\delta'$ (while
noting we've assumed
$\spec`\big(\sigma_{EX}^{(i)}) = \spec`\big(\sigma_{EX}^{(j)})$).  Then
\begin{align}
  \abs[\Big]{
  `\big[S(\bar\mu) - S(\bar\lambda)]
  - 
  `\big[S(\bar\mu') - S(\bar\lambda')]
  }
  &\geq
    \abs[\big]{S(\bar\mu) - S(\bar\mu')}
    - \abs[\big]{S(\bar\lambda) - S(\bar\lambda')}
    \nonumber\\
  &\geq 
    \abs[\big]{ S`\big(\sigma_X^{(i)}) - S`\big(\sigma_X^{(j)}) }
    - 4`\Big(h(\delta') + \delta'\log(d_X-1))
    \nonumber\\
  &\geq  \delta  - 4`\Big(h(\delta') + \delta'\log(d_X-1))\ ,
\end{align}
making use of the Fannes-Audenaert inequality, where
$h(p) = -p\log(p) - `(1-p)\log(1-p)$ is the binary entropy, and
recalling~\eqref{eq:f8eiujdskljdaofi9io}.  If we choose $\delta'>0$ sufficiently
small such that $\delta > 4`\big[h(\delta') + \delta'\log`(d_X-1)]$, there is no
$w\in\mathbb{R}$, let alone $w\in\mathscr{W}$, such that
$`\big[S(\bar\mu) - S(\bar\lambda)] = \beta w = `\big[S(\bar\mu') -
S(\bar\lambda')]$.  Therefore, if we fix the choice of $\mu,\mu'$ given above,
there is no term in the sum in~\eqref{eq:fdkjhar8932ioawdsjk} with these
$\mu,\mu'$; we find
\begin{align}
  \Pi_{X^n}^\mu 
  \tr_{E^n}`\Big{
    \text{\eqref{eq:afijegqw9dfkfdnq9id}} - \Delta
  }\,
  \Pi_{X^n}^{\mu'}
  =  0\ .
\end{align}
Define the following projector:
\begin{align}
  M^i_{X^n} = \sum_{\substack{
  \mu:\\ \onenorm{\bar\mu - \spec(\sigma_{X}^{(i)})} \leq\delta'
  }} \Pi_{X^n}^\mu\ .
\end{align}
We then find,
\begin{align}
  \onenorm[\bigg]{
  M^i_{X^n}
  \frac1{\sqrt{p_ip_j}} \bra{i}_{R'}\, \tilde\rho_{X^n\bar{R}^nR'}\,\ket{j}_{R'}
  M^j_{X^n}
  }
  \leq \poly(n)\,\exp`\big{-n\eta}\ .
  \label{eq:5jbvdiuowafnjdsafksanfk}
\end{align}
On the other hand, making use of~\eqref{eq:ewu90ofjknlsadosakl} we have that
\begin{align}
  \bra{\sigma^{(i)}}_{EX\bar R}^{\otimes n} \,
  M^i_{X^n} \, \ket{\sigma^{(i)}}_{EX\bar R}^{\otimes n}
  \;  &=\;   \tr`\big[ M^i_{X^n} \sigma^{(i)\,\otimes n}_{EX\bar R} ]
        \nonumber\\
  \ &= \ 1\  - \sum_{\substack{\mu:\\
  \onenorm{\bar\mu - \spec(\sigma_{X}^{(i)})} > \delta'
  }} \tr`\big[ \Pi_{X^n}^{\mu} \sigma^{(i)\,\otimes n}_{X^n} ]
  \; \geq \;
  1 - \poly`(n)\,\ee^{-n\eta}\,.
\end{align}
This implies
\begin{align}
  M^i_{X^n} \, \ket{\sigma^{(i)}}_{EX\bar R}^{\otimes n}
  &=
  \ket{\sigma^{(i)}}_{EX\bar R}^{\otimes n}
  +
  \ket{\epsilon^{(i)}}_{E^nX^n\bar R^n}\ ;
  &
    \ket{\epsilon^{(i)}}_{E^nX^n\bar R^n}
    &= `(\Ident - M^i_{X^n}) \, \ket{\sigma^{(i)}}_{EX\bar R}^{\otimes n}\ ,
\end{align}
noting that $`(\Ident - M^i_{X^n})$ is a projector, and with
\begin{align}
  \norm[\big]{\ket{\epsilon^{(i)}}_{E^nX^n\bar R^n}}^2
  = \braket{\epsilon^{(i)}}{\epsilon^{(i)}}_{E^nX^n\bar R^n}
  = 
  \bra{\sigma^{(i)\,\otimes n}}_{EX\bar R} \,
  `(\Ident - M_{X^n}^i)
  \, \ket{\sigma^{(i)\,\otimes n}}_{EX\bar R}
  \leq \poly(n)\,\ee^{-n\eta}\ .
\end{align}
Recalling~\eqref{eq:fokdih0q9wiadokfjdi}, we find for the
state without the dephasing:
\begin{align}
  \hspace*{2em}&\hspace*{-2em}
  \onenorm[\bigg]{
  M^i_{X^n}
  \frac1{\sqrt{p_ip_j}} \bra{i}_{R'}\, \rho_{X^n\bar{R}^nR'}\,\ket{j}_{R'}
  M^j_{X^n}
  }
  = 
  \onenorm[\Big]{
    M^i_{X^n}
    \tr_{E^n}`\Big{ \ketbra[\big]{\sigma^{(i)}}{\sigma^{(j)}}_{EX\bar R}^{\otimes n} } \,
    M^j_{X^n}
  }
    \nonumber\\
  &= 
  \onenorm[\Big]{
    \tr_{E^n}`\Big{
    `\Big( \ket[\big]{\sigma^{(i)}}_{EX\bar R}^{\otimes n} + \ket{\epsilon^{(i)}} )
    `\Big( \bra[\big]{\sigma^{(j)}}_{EX\bar R}^{\otimes n} + \bra{\epsilon^{(j)}} ) \,
  } }
\nonumber\\
  &\geq  \onenorm[\Big]{
    \tr_{E^n}`\Big{ \ketbra[\big]{\sigma^{(i)}}{\sigma^{(j)}}_{EX\bar R}^{\otimes n} } \,
  } -  \onenorm[\Big]{
    \tr_{E^n}`\Big{ \ket[\big]{\sigma^{(i)}}^{\otimes n}\bra[\big]{\epsilon^{(j)}} } \,
  } - \onenorm[\Big]{
    \tr_{E^n}`\Big{ \ket[\big]{\sigma^{(j)}}^{\otimes n}\bra[\big]{\epsilon^{(i)}} } \,
  } -  \onenorm[\Big]{
    \tr_{E^n}`\Big{ \ket[\big]{\epsilon^{(i)}}\bra[\big]{\epsilon^{(j)}} } \,
  }
\nonumber\\
  &\geq `\big( \onenorm[\big]{
    \tr_{E^n}`\big{ \ketbra[\big]{\sigma^{(i)}}{\sigma^{(j)}}_{EX\bar R} } \,
  } )^n - \poly(n)\,\ee^{-n\eta/2}\ ,
    \label{eq:diohe2dksjiue9oqfs90dask}
\end{align}
where the last step follows from the fact that for any
$\ket\psi_{AB},\ket\epsilon_{AB}$ with $\braket\psi\psi\leq 1$, we have
$\onenorm[\big]{\tr_A`(\ketbra\psi\epsilon)}^2 = %
\max_{U_B} \abs[\big]{ \tr`\big[U_B\,\ketbra\psi\epsilon_{AB}] }^2
= 
\max_{U_B} \bra{\epsilon}_{AB} U_B\,\proj\psi U_B^\dagger \ket{\epsilon}_{AB}
\leq \braket\epsilon\epsilon = \norm{\ket\epsilon}^2$.
\end{widetext}

An example with the desired properties is then constructed by finding states
$`{ \ket{\sigma^{(i)}} }$ that satisfy~\eqref{eq:f8eiujdskljdaofi9io} along with
$\onenorm[\big]{ \tr_{E}`\big{ \ketbra[\big]{\sigma^{(i)}}{\sigma^{(j)}}_{EX\bar
      R} } } = 1$ for some $i\neq j$.  In this case, the
$\bra{i}_{R'}(\cdot)\ket{j}_{R'}$ off-diagonal block of $\rho_{X^n\bar R^n R'}$,
in the $M_{X^n}^i(\cdot) M_{X^n}^j$ submatrix, remains close to one; yet the
corresponding block in $\tilde\rho_{X^n\bar R^n R'}$ vanishes exponentially in
$n$, as per~\eqref{eq:5jbvdiuowafnjdsafksanfk}.

A possible choice is the following.
We choose
$E\simeq X\simeq \bar R$ as one qubit, and we pick two states with
$p_1=p_2 = 1/2$:
\begin{align}
  \sigma^{(1)}_{EX\bar R}
  &= \ket{GHZ}_{EX\bar R} = \frac1{\sqrt{2}}`\big[\ket{000} + \ket{111}]_{EX\bar R}\ ;
    \nonumber\\
  \sigma^{(2)}_{EX\bar R}
  &= (\ket{\phi^+}_{E\bar R})\ket0_X
    = \frac1{\sqrt{2}}`\big[\ket{000} + \ket{101}]_{EX\bar R}\ ,
\end{align}
where $\ket{\phi^+} =`[\ket{00}+\ket{11}]/\sqrt2$.
Then
\begin{align}
  \begin{aligned}
  \sigma_{EX}^{(1)} &= \frac12\proj{00}_{EX} + \frac12\proj{11}_{EX}\ ;
  \\
  \sigma_{EX}^{(2)} &= \frac{\Ident_E}{2}\otimes\proj0_X\ ,
  \end{aligned}
\end{align}
noting that
\begin{align}
  \spec`\big(\sigma_{EX}^{(1)}) =   \spec`\big(\sigma_{EX}^{(2)}) = `(1/2, 1/2, 0, 0)\ .
\end{align}
Also,
\begin{align}
  \begin{aligned}
  \sigma_{X}^{(1)} &= \frac12\Ident_X\ ;
    &
  \sigma_{X}^{(2)} &= \proj0_X\ ,
  \end{aligned}
\end{align}
which ensures that
\begin{align}
  \spec`\big(\sigma_{X}^{(1)})
  &= `(1/2,1/2)\ ;
    &
  \spec`\big(\sigma_{X}^{(2)})
  &= `(1, 0)\ ,
\end{align}
so that we can choose $\delta=\log(2)$.
We can then inspect
\begin{align}
  \hspace*{1em}&\hspace*{-1em}
  \tr_{E}`\Big{ \ketbra[\big]{\sigma^{(1)}}{\sigma^{(2)}}_{EX\bar R} }
  \nonumber\\
  &=
    \frac12\tr_{E}`\Big{ `\big(\ket{000}+\ket{111})`\big(\bra{000}+\bra{101})_{EX\bar R} }
    \nonumber\\
  &=
    \frac12\ketbra{00}{00}_{X\bar R} + \frac12\ketbra{11}{01}_{X\bar R}\ .
\end{align}
This expression is already a singular value decomposition for
$ \tr_{E}`\big{ \ketbra[\big]{\sigma^{(1)}}{\sigma^{(2)}}_{EX\bar R} } $, given
that $\ket{00}\perp\ket{11}$ and $\ket{00}\perp\ket{01}$, so we can read off the
nonzero singular values $(1/2, 1/2)$.  The operator one-norm being the sum of
the singular values, we find
\begin{align}
  \onenorm[\Big]{
  \tr_{E}`\big{ \ketbra[\big]{\sigma^{(1)}}{\sigma^{(2)}}_{EX\bar R} }
  } = 1\ .
  \label{eq:1kjdsnabfioewdafdkssa}
\end{align}

To conclude, we have exhibited a state $\rho_{E^nX^n R}$ such that
$\rho_{X^n R}$ differs significantly from the work-cost-decohered version
$\tilde\rho_{X^n R}$.  A major difference is exhibited in certain off-diagonal
blocks of the matrix associated with the state on $R'$ and $\mu,\mu'$ Schur-Weyl
blocks on $X^n$.  Whereas the considered off-diagonal block of
$\tilde\rho_{X^nR}$ vanishes exponentially fast asymptotically for $n\to\infty$
[\cref{eq:5jbvdiuowafnjdsafksanfk}], the same off-diagonal block in
$\rho_{X^nR}$ remains of constant norm
[\cref{eq:diohe2dksjiue9oqfs90dask,eq:1kjdsnabfioewdafdkssa}].

\vfill\hbox{}
\onecolumngrid
\vspace{1cm}
\begin{center}%
\par\noindent\hbox{\hfill\fbox{\vbox{\textbf{REFERENCES}}}\hfill}\par\end{center}
\twocolumngrid

%

\end{document}